\documentclass[11pt,a4paper]{article}

\usepackage{amssymb,amsthm,amsmath,mathtools}
\usepackage{thmtools,thm-restate}
\usepackage[round]{natbib}
\usepackage{pifont}

\usepackage[usenames,svgnames,xcdraw,table]{xcolor}
\definecolor{DarkGreen}{rgb}{0.1,0.5,0.1}
\usepackage[backref=page]{hyperref}
\hypersetup{
	colorlinks=true,
	linkcolor=red,
	urlcolor=DarkGreen,
	citecolor=blue
}
\renewcommand*{\backref}[1]{}
\renewcommand*{\backrefalt}[4]{%
    \ifcase #1 (Not cited.)%
    \or        (Cited on page~#2)%
    \else      (Cited on pages~#2)%
    \fi}
\usepackage[capitalise,noabbrev]{cleveref}
\Crefname{property}{Property}{Properties}
\Crefname{theorem}{Theorem}{Theorems}
\Crefname{example}{Example}{Examples}
\Crefname{table}{Table}{Tables}

\usepackage{cancel}
\usepackage{tabularx}
\usepackage{nicefrac}
\usepackage{enumerate}
\usepackage{etoolbox}
\usepackage[margin=1in]{geometry}
\usepackage[T1]{fontenc}
\usepackage[tt=false]{libertine}

\usepackage{url}
\usepackage[utf8]{inputenc}
\usepackage[small]{caption}
\usepackage{graphicx}
\usepackage{booktabs}
\usepackage{mathrsfs,amsfonts,dsfont,authblk}
\usepackage{array,multirow,graphicx,bigdelim}

\usepackage[linesnumbered,lined,boxed,ruled,vlined]{algorithm2e}

\SetKwInOut{Parameters}{Parameters}
\SetKwComment{Comment}{$\triangleright$\ }{}
\SetAlFnt{\small}
\SetAlCapFnt{\small}
\SetAlCapNameFnt{\small}
\SetAlCapHSkip{0pt}
\IncMargin{-\parindent}

\usepackage{tikz}
\usepackage{tikz-dimline}
\usetikzlibrary{fit,calc,shapes,shapes.misc,shapes.multipart,arrows,math,patterns,positioning}
\usepgflibrary{arrows}
\newcommand*{\tikzmk}[1]{\tikz[remember picture,overlay,] \node (#1) {};\ignorespaces}
\newcommand{\boxit}[1]{\tikz[remember picture,overlay]{\node[yshift=3pt,xshift=4pt,fill=#1,opacity=.25,fit={(A)($(B)+(1.0\linewidth,.8\baselineskip)$)}] {};}\ignorespaces}
\colorlet{mygray}{gray!40}

\newcommand\qedblob{\ding{113}} 
\def\literalqed{{\ \nolinebreak\hfill\mbox{\qedblob\quad}}} 
\def\qed{\literalqed} 
\newenvironment{proofsketch}{\par\addvspace{\baselineskip}\noindent{\emph{Proof Sketch.}}}{\literalqed\bigskip}

\makeatletter
\renewcommand{\paragraph}{%
 \@startsection{paragraph}{4}%
 {\z@}{1.1ex \@plus 1ex \@minus .2ex}{-1em}%
 {\normalfont\normalsize\bfseries}%
}
\def\thickhline{%
  \noalign{\ifnum0=`}\fi\hrule \@height \thickarrayrulewidth \futurelet
   \reserved@a\@xthickhline}
\def\@xthickhline{\ifx\reserved@a\thickhline
               \vskip\doublerulesep
               \vskip-\thickarrayrulewidth
             \fi
      \ifnum0=`{\fi}}
\makeatother
\newlength{\thickarrayrulewidth}
\makeatletter
\let\oldnl\nl
\newcommand{\nonl}{\renewcommand{\nl}{\let\nl\oldnl}}
\makeatother

  {\list{}{\leftmargin=#1\rightmargin=#1}\item[]}%
  {\endlist}

\usepackage[labelfont={normalfont,bf},textfont=it]{caption}
\usepackage{subcaption}

\theoremstyle{definition}

\newtheorem{definition}{Definition}

\newtheorem{examplex}{Example}
\newenvironment{example}{\pushQED{\qed}\examplex}{\popQED\endexamplex}
\theoremstyle{remark}

\usepackage{flushend}
\usepackage[utf8]{inputenc}
\Crefname{claim}{Claim}{Claims}

\allowdisplaybreaks

\newcommand{\C}{\mathcal{C}}

\newcommand{\eps}{\varepsilon}
\newcommand{\I}{\mathcal{I}}

\newcommand{\myoverbar}[1]{\makebox[0pt]{$\phantom{#1}\overline{\phantom{#1}}$}#1}
\newcommand{\myunderbar}[1]{\makebox[0pt]{$\phantom{#1}\underline{\phantom{#1}}$}#1}

\renewcommand{\O}{\mathcal{O}}
\newcommand{\opt}{\texttt{opt}}

\renewcommand{\P}{\mathcal{P}}

\newcommand{\ProxyVoting}{\textup{\textsc{Proxy Voting}}}
\newcommand{\R}{\mathbb{R}}
\renewcommand{\top}{\texttt{top}}

\begin{document}

\title{Representative Proxy Voting}
\date{}

\author[1]{Elliot Anshelevich}
\author[2]{Zack Fitzsimmons}
\author[3]{Rohit Vaish}
\author[4]{Lirong Xia}
\affil[1]{Rensselaer Polytechnic Institute\\
	{\small\texttt{eanshel@rpi.edu}}}
\affil[2]{College of the Holy Cross\\
	{\small\texttt{zfitzsim@holycross.edu}}}
\affil[3]{Tata Institute of Fundamental Research\\
	{\small\texttt{rohit.vaish@tifr.res.in}}}
\affil[4]{Rensselaer Polytechnic Institute\\
	{\small\texttt{xial@cs.rpi.edu}}}

\maketitle

\begin{abstract}
We study a model of proxy voting where the candidates, voters, and proxies are all located on the real line, and instead of voting directly, each voter delegates its vote to the closest proxy. The goal is to find a set of proxies that is \emph{$\theta$-representative}, which entails that for any voter located anywhere on the line, its favorite candidate is within a distance $\theta$ of the favorite candidate of its closest proxy. This property guarantees a strong form of representation as the set of voters is not required to be fixed in advance, or even be finite. We show that for candidates located on a line, an optimal proxy arrangement can be computed in polynomial time. Moreover, we provide upper and lower bounds on the number of proxies required to form a $\theta$-representative set, thus showing that a relatively small number of proxies is enough to capture the preferences of any set of voters. An additional beneficial property of a $\theta$-representative proxy arrangement is that for strict-Condorcet voting rules, the outcome of proxy voting is similarly close to the outcome of direct voting.
\end{abstract}

\section{Introduction}

It is natural to consider settings where voters either may not be able to or may not be willing to directly  cast their vote, but instead decide to delegate their votes to a proxy. In much of the related work on proxy voting, the proxies are chosen from the set of eligible voters to then represent the electorate (see, e.g.,~\citealp{coh-man-mei-mei-ord:c:proxy-better-outcomes}). In this paper, we consider a model for proxy voting that \emph{introduces} proxies to the election using only the arrangement of the candidates and a given distance $\theta$ for {\em any} collection of voters.

On first glance, the idea of ``creating'' proxies from scratch might seem unnatural. However, note that in real-world applications where the `candidates' correspond to different resource allocations or long-term policy decisions (not actual political candidates), and the `proxies' correspond to human representatives/experts or policy positions, it is reasonable to assume that human voters would find a proxy more relatable than the actual candidate.

We call our arrangement \emph{$\theta$-representative} since each voter's closest proxy is guaranteed to have a top preference that is within $\theta$ of the voter's. This can be interpreted as providing a set of \emph{allowed votes} identified by a set of proxies placed in the metric space, and naturally the voters cast the vote of the proxy closest to them with the guarantee that this preference is {close to} the voter's. Another way to think about our model is that we form a set of \emph{representatives} (proxies) whose choice does not depend on the locations of the voters (since these locations are difficult to determine exactly, or perhaps they are not static and change over time). The goal for this set of representatives is that it captures well the opinions of the voters, even as the voters change over time, if the voters are able to express their opinions by delegating their vote to their closest representative (proxy).

We consider elections where the voters, the proxies, and the candidates are all located in a metric space, and each voter and each proxy has spatial preferences determined by its distance to each candidate (see, e.g.,~\citealp{Sch08} for a survey on spatial voting). We focus our results on the one-dimensional case where all candidates and voters lie in the interval $[0,1]$ (as in \citealp{coh-man-mei-mei-ord:c:proxy-better-outcomes}). While not as general as an arbitrary metric space, it is an important step in understanding the behavior of the problem in more general domains. The one-dimensional assumption encompasses the well-studied domains of single-peaked~\citep{bla:j:rationale-of-group-decision-making} and single-crossing preferences~\citep{mir:j:single-crossing}.

What sets our contributions apart from the related work on proxy voting is that we determine a proxy arrangement \emph{without} knowing the locations of the voters; that is, a $\theta$-representative proxy
arrangement is representative of {\em all} possible sets of voters simultaneously. In contrast, related work on proxy voting generally selects proxies or representatives from among a given set of voters~(see, e.g.,~\citealp{coh-man-mei-mei-ord:c:proxy-better-outcomes, mei-san-ten:t:representative-committees} and the references therein).

\renewcommand{\arraystretch}{1.5}
 \begin{table*}[t]
 \centering
 \footnotesize
 \begin{tabular}{
 |>{\centering}m{0.27\textwidth}|
 >{\centering}m{0.22\textwidth}|
 >{\centering}m{0.22\textwidth}|
 >{\centering\arraybackslash}m{0.14\textwidth}|}
     %
     \hline
     \multirow{2}{*}{\textbf{Restrictions on proxies} $\downarrow$} & \multicolumn{2}{c|}{\textbf{Bounds on the number of proxies}} & \textbf{Computational}\\
     \cline{2-3}
     & Upper bound & Lower bound & \textbf{Results}\\
     \hline
     \multirow{2}{*}{On top of candidates (Restricted)} & \normalsize{$2\lfloor\frac{1}{\theta}\rfloor$}  & \normalsize{$2\lfloor\frac{1}{\theta}\rfloor$}  & $\opt$ in poly time\\
     & (\Cref{thm:UpperBound_Restricted}) & (\Cref{prop:LowerBound_Restricted}) & (\Cref{thm:Opt_Restricted})\\
     \cline{1-4}
     \multirow{2}{*}{Anywhere within $\mathbb{R}$ (Unrestricted)} & \normalsize{$\frac{3}{2}\lceil\frac{1}{\theta}\rceil$}  & \normalsize{$\lceil\frac{1}{\theta}\rceil$} & $\opt$ in poly time\\
     & (\Cref{thm:UpperBound_Unrestricted}) & (\Cref{prop:LowerBound_Unrestricted}) & (\Cref{thm:Unrestricted_Proxies_Optimal})\\
     \hline
 \end{tabular}
 \caption{Summary of results. For each assumption on the positioning of the proxies (left column), the second and the third columns provide upper and lower bounds, respectively, on the number of proxies as a function of $\theta$ for a $\theta$-representative proxy arrangement. The rightmost column contains the computational results for optimal proxy arrangements.}
 \label{tab:Results}
 \end{table*}

\paragraph{Our Contributions}
Our main contributions are as follows (see also 
Table~\ref{tab:Results}):

\begin{itemize} 
    \item We introduce a new model of proxy voting and measure for
    the representation of voter preferences. We consider both the \emph{unrestricted} case where the proxies can be placed anywhere, and the \emph{restricted} case where the proxies can only be placed at the candidate locations.\footnote{Our results also hold if the possible proxy locations are more general, as long as they include all candidate locations as a subset.}
    
    \item We first provide algorithms for computing an optimal $\theta$-representative proxy arrangement (i.e., one that uses a \emph{minimum} number of proxies) in polynomial time~(\Cref{thm:Opt_Restricted,thm:Unrestricted_Proxies_Optimal}). However, these algorithms do not provide much insight into \emph{how many} proxies are actually necessary in order to be $\theta$-representative, for any possible set of candidates and voters. Because of this, we also prove upper and lower bounds on the number of proxies needed to have a $\theta$-representative proxy arrangement (see Table~\ref{tab:Results}). These bounds show that a relatively small number of proxies is enough to capture the preferences of any set of voters, even in the worst case.
    
    \item Our results also address the \emph{dual} problem of minimizing the distance threshold $\theta$ for a given number of proxies. For example, we prove that for the unrestricted case, one can use a given budget of $k$ proxies to compute a $\theta$-representative arrangement with $\theta\le \frac{3}{2}\lceil\frac{1}{k}\rceil$ (\Cref{cor:UpperBound_Unrestricted_Theta}).
    
    \item We observe that the proxy arrangements determined by our algorithms are not only $\theta$-representative with respect to the voters, but for elections using strict-Condorcet voting rules, the direct election outcome (i.e., without proxies) is within $\theta$ of the proxy-voting outcome (\Cref{prop:rep-voting}).
\end{itemize}

\section{Related Work}

\citet{alg:j:proxy-voting} introduce a general proxy voting model where there is a fixed set of proxies and the voters can choose a proxy to represent them. \citet{coh-man-mei-mei-ord:c:proxy-better-outcomes} consider a proxy voting setting closely related to ours where the voters, the proxies, and the candidates are on an interval, and the voters delegate their vote to their nearest proxy. However, in their work, the proxies are selected at random and the focus is on how the outcome is affected.

\citet{gre:j:proxy-direct} consider proxy voting with spatial preferences where proxies are elected from among the voters, but additionally explore settings where proxies can further delegate their vote, which is often referred to as delegative democracy (see, e.g.,~\citealp{kah-mac-pro:c:liquid-democracy-algo,goe-kah-mac-pro:c:fuild-mech-democracy}). More general models of delegative democracy have also been studied (see, e.g.,~\citealp{bri-tal:c:pairwise-democracy,abr-mat:c:flexible-democracy}). Another line of research considers electing a set of representatives from the voters, and generally comparing the outcome from the vote of the committee and the outcome from direct voting~\citep{sko:j:comittees-multiwinner,piv-soh:j:weighted-democracy,mei-san-ten:t:representative-committees,Magdon-Ismail2018:A-Mathematical}.

Besides having a $\theta$-representative proxy arrangement, we observe that for Condorcet-consistent voting rules, our results guarantee that the proxy-voting outcome is within $\theta$ of the direct-voting outcome. This notion of a $\theta$-representative outcome is related to, but quite different than, the notion of
distortion~\citet{pro-ros:c:distortion}. A voting rule's distortion is the worst-case ratio of the distance from the voters to the winner and the candidate that minimizes this distance. The study of distortion is an active area of research in voting (see e.g., \citealp{ans-bha-elk-pos-sko:j:social-choice-metric,Goel17,abr-ans-zhu:c:awareness-metric,kem:c:communication-and-distortion}).

\citet{FMS+20Designing} study a similar model to ours for a problem motivated by fairness in academic hiring. They consider a setup where a set of applicants (analogous to candidates in our model) are chosen by a set of experts (analogous to proxies). Each applicant is associated with a known quality score, and the applicants as well as the experts are arranged on a line. Each expert votes for its closest candidate, where closeness is defined as ``distance minus quality.'' As with the related work on distortion, the aim of~\citet{FMS+20Designing} is related to the outcome, while our paper is focused on finding a representative set of proxies.

Finally, we note that our model bears some resemblance to the $p$-coverage version of the facility location problem~\citep{HT91improved}. Under this problem, we are given a set $D = \{v_1,\dots,v_n\}$ of $n$ demand points, and the goal is to select a subset $S \subseteq D$ of at most $p$ points such that each point in $D$ is at most a fixed distance away from some point in $S$. If we think of the demand set $D$ as the set of candidates and the set $S$ as the proxies, then the restricted variant of our problem is analogous to the $p$-coverage facility location problem. However, note that the unrestricted variant of our problem is significantly more general.

\section{Preliminaries}

For any natural number $s \in \mathbb{N}$, let $[s] \coloneqq \{1,\dots,s\}$. 
%
Our model involves three types of entities---\emph{candidates}, \emph{voters}, and \emph{proxies}---which are described below.

\paragraph{Candidates:} Let $\C = \{c_1,\dots,c_m\}$ denote a set of $m \in \mathbb{N}$ candidates that are arranged on a line segment $[0,1]$. We will overload notation to denote the position of the $i^\text{th}$ candidate also by $c_i \in [0,1]$. We will assume throughout that the extreme candidates are located at the endpoints of $[0,1]$ (i.e., $c_1 = 0$ and $c_m = 1$), and that all candidates have distinct locations (i.e., for any distinct $i,j \in [m]$, $c_i \neq c_j$).

For any pair of adjacent candidates $c_i$ and $c_{i+1}$, their \emph{candidate bisector} is the vertical line at $(c_i + c_{i+1})/2$. Notice that with $m$ candidates, there can, in general, be $\binom{m}{2}$ different bisectors. However, unless stated otherwise, the term `candidate bisector' will refer to a bisector between \emph{adjacent} candidates.

\paragraph{Voters:} A voter can be located anywhere in $[0,1]$ and is identified by its location. For any voter $v \in [0,1]$, its favorite candidate, denoted by $\top(v)$, is the candidate that is closest to it. That is, $\top(v) \coloneqq \arg\min_{c_i \in \C} |v - c_i|$, where ties are broken according to any fixed \emph{directionally consistent} tie-breaking rule (see \Cref{defn:Tie-Breaking}). Note that we do not assume the set of voters to be fixed in advance. This is because our results apply to any arbitrary collection of voters that could be located anywhere within $[0,1]$.

\begin{definition}[\textbf{Directionally consistent tie-breaking rule}]
A tie-breaking rule is a function $\tau: [0,1]^3 \rightarrow [0,1]$ that maps any triple $v,x,y \in [0,1]$ as follows:
$$\tau(v,x,y) =\begin{cases} 
x &\mbox{if } |v-x|<|v-y| \\
y & \mbox{if } |v-x|>|v-y| \\
\text{either $x$ or $y$} & \mbox{otherwise}.
\end{cases}$$
That is, a tie-breaking rule always maps to the point that is closer to $v$, and in case of a tie, picks exactly one of the two points. We say that a tie-breaking rule $\tau$ is \emph{directionally consistent} if, for any fixed $v \in [0,1]$, either $\tau$ always tie-breaks to the left of $v$ or always to the right of $v$ (the choice of direction could depend on $v$). That is, for any fixed $v \in [0,1]$, either $\tau(v,x,y) = x$ for all $x,y \in [0,1]$ such that $x \leq v \leq y$ and $|v-x|=|v-y|$, or $\tau(v,x,y) = y$ for all $x,y \in [0,1]$ such that $x \leq v \leq y$ and $|v-x|=|v-y|$.
\label{defn:Tie-Breaking}
\end{definition}

For any candidate $c_i \in \C$, its \emph{Voronoi cell} $V_i$ denotes the set of all voter locations $v \in [0,1]$ whose favorite candidate is $c_i$, i.e., $V_i \coloneqq \{v \in [0,1]: \top(v) = c_i \}$. Notice that the Voronoi cells $V_1,\dots,V_m$ induce a partition of the line segment $[0,1]$.

\paragraph{Proxies:} Our model also includes a finite set $\P$ of proxies whose role is to vote on behalf of the voters. Specifically, each voter $v$ delegates its vote to its \emph{nearest proxy}, denoted by $p^v\coloneqq \arg\min_{p \in \P} |v - p|$. Each proxy $p \in \R$ then votes for its favorite candidate, denoted by $\top(p)$, which is defined as the candidate closest to it, i.e., $\top(p) \coloneqq \arg\min_{c_i \in \C} |p-c_i|$. As before, ties are broken according to a directionally consistent tie-breaking rule (see \Cref{defn:Tie-Breaking}). We will often use the term \emph{proxy arrangement} to refer to a set of proxies placed on the real line.

\paragraph{Representative proxy arrangements:}
We will now formally define what it means for a proxy arrangement to be \emph{representative}. The definition is stated in terms of a parameter $\theta \in [0,1]$ that corresponds to a distance threshold. We will find it convenient to call two points $x,y \in [0,1]$ to be \emph{$\theta$-close} if $|x-y| \leq \theta$, and call them \emph{$\theta$-far} otherwise. 

Given any $\theta \in [0,1]$, we say that a proxy arrangement is \emph{$\theta$-representative} if for \emph{every} voter location, the favorite candidate of the voter is $\theta$-close to the favorite candidate of its closest proxy.\footnote{One might ask whether, in place of an additive approximation in the definition of $\theta$-representation, a multiplicative guarantee could be used instead. We note that a multiplicative approximation would implicitly use different ``$\theta$'' values for different voters, i.e., if a voter's favorite candidate is at position $z$, then its closest proxy's favorite should be within $[(1-\eps)\cdot z,(1+\eps)\cdot z]$. Thus, the voters whose favorite candidate is far away from $x=0$ have a greater ``slack'' in their representation (and are therefore loosely represented), whereas the voters whose favorite candidates are closer to $x=0$ enjoy a stronger representation. While this notion is mathematically well-defined, it does not seem as natural as our additive notion since the representation guarantee is no longer uniform for all voters.}

\begin{definition}[\textbf{$\theta$-representative proxy arrangement}]
An arrangement of proxies is \emph{$\theta$-representative} if for any voter $v \in [0,1]$, its favorite candidate is $\theta$-close to the favorite candidate of its nearest proxy. That is, for every $v \in [0,1]$, $|\top(v) - \top(p^{v})| \leq \theta$.
\label{defn:Theta_representative}
\end{definition}

This property essentially says that for every voter, no matter where they may be located, it must be that their preference is not too different from the preference of their closest proxy. Therefore, the voter should feel reasonably satisfied with the proxy arrangement, as their closest proxy (to whom they yield the power of their vote) will be somewhat representative of their interests. On the other hand, if an arrangement is {\em not} $\theta$-representative for some large $\theta$, this means that there are collections of voters which would be unhappy with this proxy arrangement, as the votes of the proxies would heavily disagree with the preferences of the voters.

We will now define the central problem studied in this paper.

\begin{definition}[\textbf{\ProxyVoting{}}]
An instance of the \ProxyVoting{} problem $\I = \langle \C, \theta \rangle$ is specified by a set of candidates $\C$ in $[0,1]$ and a parameter $\theta \in (0,1)$. The goal is to compute an arrangement of proxies that is \emph{$\theta$-representative} for every location $v \in [0,1]$.\footnote{We exclude the degenerate corner cases of $\theta = 0$ and $\theta=1$ from the definition. Indeed, if $\theta=0$, then it is easy to see that an optimal proxy arrangement requires as many proxies as candidates (realized by placing a proxy on each candidate). On the other hand, $\theta=1$ involves placing a single proxy anywhere in $[0,1]$.}
\label{defn:Proxy_Voting}
\end{definition}

We once again stress that a $\theta$-representative proxy arrangement should be representative of {\em all possible sets of voters}: When choosing appropriate proxies to represent the populace, we do not have information about the voter locations; we only know the candidate locations. No matter where the voters are located, or how they change in the future, the proxies will still be representative of their views.

It is worth pointing out that the assumption about the candidates and voters lying inside $[0,1]$ is without loss of generality. Indeed, by assuming the extreme candidates to be at $x=0$ and $x=1$, we are able to specify the parameter $\theta$ in absolute terms. Otherwise, we would need to define the distance threshold as ``$\theta$ times the distance between the extreme candidates''. The voters too, in principle, can be anywhere on the real line. However, note that for any voter located (weakly) to the left of $x=0$, its favorite candidate is at $x=0$. Furthermore, it can be shown that under an optimal proxy arrangement, the favorite proxy of any voter to the left of $x=0$ is, without loss of generality, also the favorite proxy of the voter \emph{at} $x=0$ (in other words, the leftmost proxy bisector is weakly to the right of $x=0$). Thus, $\theta$-representation for \emph{any} voter to the left of $x=0$ is subsumed by the same condition for the voter at $x=0$, and therefore, it suffices to assume that there are no voters in $x<0$. A similar argument can be made for voters to the right of $x=1$.

We will study two variants of \ProxyVoting{} which we call the \emph{restricted} and \emph{unrestricted} versions. Under the \emph{restricted} version of the problem, we require that the set of proxies must be a subset of candidate locations, that is, the proxies must lie on top of the candidates. In the \emph{unrestricted} version, the proxies can lie anywhere on the real line.

Note that the computational problem pertaining to \Cref{defn:Proxy_Voting} can be formalized as a decision as well as an optimization problem. The decision version asks whether, given an instance $\I$ and a natural number $k \in \mathbb{N}$, there exists a $\theta$-representative proxy arrangement for $\I$ consisting of at most $k$ proxies. The optimization version asks whether, given an instance $\I$, the \emph{optimal} $\theta$-representative proxy arrangement for $\I$ can be computed in polynomial time. An arrangement of $k$ proxies is \emph{optimal} for instance $\I$ if there is no other arrangement of $k-1$ or fewer proxies that is $\theta$-representative. We will write $\opt(\I)$ to denote the number of proxies in any optimal $\theta$-representative arrangement for the instance $\I$. To make the computational problem meaningful, we will assume throughout that all candidate locations in $\C$ and the parameter $\theta$ are rational.

\subsection*{Basic Properties and Examples}
Let us start by discussing some of the issues that arise when constructing $\theta$-representative sets of proxies, and attempt to build intuition about our techniques. First, consider the following lemma, which points out a precise relationship between candidate bisectors and proxy bisectors.

\begin{restatable}{lemma}{BisectorsNonCoincident}
Let $\I = \langle C,\theta \rangle$ denote an instance with a pair of adjacent candidates $c_i$ and $c_{i+1}$ that are $\theta$-far. Let $\P$ be any proxy arrangement such that none of the bisectors between adjacent proxies in $\P$ coincides with the bisector between $c_i$ and $c_{i+1}$. Then, $\P$
is not $\theta$-representative for $\I$.
\label{lem:Bisectors_Non_Coincident}
\end{restatable}
\begin{proof}
Let $p_1,\dots,p_k$ denote the proxies in $\P$. Among all proxy bisectors between adjacent proxies, let the one between $p_j$ and $p_{j+1}$ be closest to the candidate bisector between $c_i$ and $c_{i+1}$ (see \Cref{fig:Non_Coincident_Bisectors}). Without loss of generality, we can assume that this proxy bisector is to the right of the candidate bisector, i.e., $\frac{c_i+c_{i+1}}{2} < \frac{p_j+p_{j+1}}{2}$. Let $d \coloneqq \min\{\frac{p_j+p_{j+1}}{2}, c_{i+1}\} - \frac{c_i+c_{i+1}}{2}$ and observe that $d>0$.

\begin{figure}[h]
\centering
\begin{tikzpicture}[line width=0.8pt]
	\tikzset{candidate/.style = {draw,shape=rectangle,scale=1.1}}
	\tikzset{voter/.style = {circle, fill=blue, minimum size=4pt, inner sep=0pt, outer sep=0pt}}
	\tikzmath{\proxyYoffset=-0.3;}
		\draw[help lines, dashed, step = 0.5cm] (-3, -0.45) grid (3, 0.45);
		\node[candidate,label=above:{$c_i$}] at (-2,0){};
		\node[candidate,label=above:{$c_{i+1}$}] at (2,0) {};
	    \draw[line width=1.2pt] (0, -0.4) -- (0, 0.4);
	    \draw[line width=1.2pt, red, dashed] (1, -0.4) -- (1, 0.4);
		\node[voter,label=above:{$v^R$}] at (0.5,0) {};
		\node[voter,label=above:{$v^L$}] at (-0.5,0) {};
	    \dimline[extension start length=0,extension end length=0] {(-2,-0.5)}{(2,-0.5)}{$>\theta$};
\end{tikzpicture}
\caption{Failure of $\theta$-representation when none of the proxy bisectors between adjacent proxies coincides with the candidate bisector between $\theta$-far candidates. The candidate bisector is shown as solid vertical line, and its closest proxy bisector is shown as a dashed vertical line in red.}
\label{fig:Non_Coincident_Bisectors}
\end{figure}
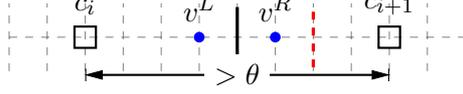

Consider a pair of voters $v^L \in (\frac{c_i+c_{i+1}}{2} - d, \frac{c_i+c_{i+1}}{2})$ and $v^R \in (\frac{c_i+c_{i+1}}{2}, \frac{c_i+c_{i+1}}{2} + d)$. Notice that the favorite candidates of $v^L$ and $v^R$ are $c_i$ and $c_{i+1}$, respectively. Furthermore, since the closest proxy bisector is at least a distance $d$ away from the candidate bisector, the two voters must have the same closest proxy, i.e., $p^{v^L} = p^{v^R} = p_j$. Then, regardless of the favorite candidate of $p_j$, $\theta$-representation must be violated for at least one of $v^L$ or $v^R$.
\end{proof}

It is relevant to note that the implication in \Cref{lem:Bisectors_Non_Coincident} holds under both restricted and unrestricted positioning assumptions. Because of \Cref{lem:Bisectors_Non_Coincident}, it is easy to see that attempting something trivial like placing proxies at equal distances will not result in a $\theta$-representative arrangement, as shown by \Cref{eg:Evenly_Spaced} below. Essentially, the example uses the observation that the proxy bisectors under an equidistant proxy arrangement can only occur in certain fixed locations. Thus, by adversarially placing a pair of $\theta$-far candidates in a way that their candidate bisector does not coincide with any of the proxy bisectors, we can use \Cref{lem:Bisectors_Non_Coincident} to demonstrate a violation of $\theta$-representation. This further motivates the need for sophisticated algorithms for computing proxy arrangements, such as those discussed in upcoming sections.

\begin{example}[\textbf{Evenly spaced proxies may not be $\theta$-representative}]
Suppose we are given some $\theta \in (0,1)$ and a budget of $k \in \mathbb{N}$ proxies. 
We will construct an instance where evenly spacing these $k$ proxies, i.e., placing the proxies at $\frac{\ell}{k-1}$ for every $\ell \in \{0,1,\dots,k-1\}$, fails to be $\theta$-representative (we will assume, without loss of generality, that $k \geq 2$). Notice that the bisectors between adjacent proxies are located at $\frac{2\ell-1}{2(k-1)}$ for every $\ell \in [k-1]$.

Consider a sufficiently small $\eps > 0$ such that $[\theta,\theta+\eps] \subseteq [\frac{\ell-1}{k-1},\frac{\ell}{k-1}]$ for some $\ell \in [k-1]$. Note that such a choice of $\eps$ must exist since $\theta \in (0,1)$. Pick a rational point $x \in (\theta,\theta+\eps) \setminus \cup_{i \in [\ell]} \{\frac{2i-1}{k-1}\}$. 

Consider an instance with three candidates $c_1$, $c_2$, and $c_3$ that are placed at $0$, $x$, and $1$, respectively. Notice that $c_1$ and $c_2$ are $\theta$-far, and that the bisector between $c_1$ and $c_2$, which is located at $x/2$, does not coincide with any proxy bisector between adjacent proxies. Therefore, by \Cref{lem:Bisectors_Non_Coincident}, the evenly spaced proxy arrangement fails to be $\theta$-representative.
\label{eg:Evenly_Spaced}
\end{example}

We will now describe our results for restricted (\Cref{sec:Restricted_Proxies})
 and unrestricted positioning of proxies (\Cref{sec:Unrestricted_Proxies}).

\section{Restricted Positioning of Proxies}
\label{sec:Restricted_Proxies}

This section provides an algorithm for computing an optimal proxy arrangement, followed by upper and lower bounds on the number of proxies needed to be $\theta$-representative.

\subsection{Algorithm for Computing Optimal Number of Restricted Proxies}

Our first result for restricted positioning (\Cref{thm:Opt_Restricted}) shows that a $\theta$-representative arrangement with the smallest number of proxies can be computed in polynomial time. While the details are somewhat complex, the proof is via an essentially straightforward dynamic programming algorithm, and is presented in \Cref{subsubsec:Proof_Opt_Restricted} in the appendix.

\begin{restatable}[\textbf{Optimal proxy arrangement under restricted positioning}]{theorem}{RestrictedProxiesOptimal}
There is a polynomial-time algorithm that, given any instance $\langle \C, \theta \rangle$ of \ProxyVoting{} as input, terminates in polynomial time and returns an optimal $\theta$-representative proxy arrangement satisfying restricted positioning.
\label{thm:Opt_Restricted}
\end{restatable}

\subsection{Upper and Lower Bounds for Restricted Positioning of Proxies}

Although the optimum number of proxies can be computed efficiently, these algorithms do not provide any insight into \emph{how many} proxies are actually necessary in order to be $\theta$-representative for any given set of candidates and any set of voters. Because of this, we now prove upper and lower bounds on the number of proxies needed to achieve $\theta$-representation.

\begin{restatable}[\textbf{Upper bound under restricted positioning}]{theorem}{UpperBoundRestricted}
Given any instance $\langle \C, \theta \rangle$ of \ProxyVoting{}, there exists a $\theta$-representative proxy arrangement satisfying restricted positioning that consists of at most $2(\frac{1}{\theta} - 1)$ proxies if $\frac{1}{\theta} \in \mathbb{N}$, and at most $2\lfloor\frac{1}{\theta}\rfloor$ proxies otherwise. Furthermore, such an arrangement can be computed in polynomial time.
\label{thm:UpperBound_Restricted}
\end{restatable}

\begin{proofsketch} The detailed proof of \Cref{thm:UpperBound_Restricted} is presented in \Cref{subsubsec:Proof_UpperBound_Restricted} in the appendix. Here, we will describe the main steps in our algorithm by means of the example shown in \Cref{subfig:Example_UpperBound_Restricted_Instance}.

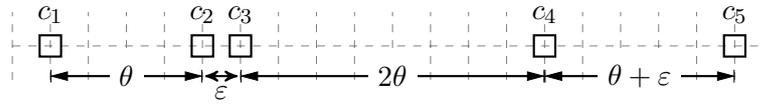
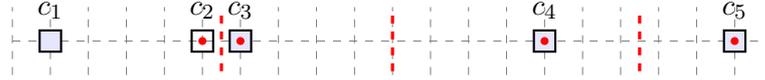
\begin{figure}[h]
\centering
\begin{subfigure}{\textwidth}
\centering
\begin{tikzpicture}[line width=0.8pt]
	\tikzset{candidate/.style = {draw,shape=rectangle,scale=1.1}}
		\draw[help lines, dashed, step = 0.5cm] (-6, -0.45) grid (4, 0.45);
		\node[candidate,label=above:{$c_1$}] at (-5.5,0){};
		\node[candidate,label=above:{$c_2$}] at (-3.5,0) {};
		\node[candidate,label=above:{$c_3$}] at (-3,0) {};
		\node[candidate,label=above:{$c_4$}] at (1,0) {};
		\node[candidate,label=above:{$c_5$}] at (3.5,0) {};
		\dimline[extension start length=0,extension end length=0]{(-5.5,-0.4)}{(-3.5,-0.4)}{$\theta$};
		\draw[>=stealth', <->, shorten <=1pt, shorten >=1pt] (-3.5,-0.4) -- node[below] {$\eps$} ++ (0.5,0);
		\dimline[extension start length=0,extension end length=0]{(-3,-0.4)}{(1,-0.4)}{$2\theta$};
		\dimline[extension start length=0,extension end length=0]{(1,-0.4)}{(3.5,-0.4)}{$\theta+\eps$};
		\node at (4.375,0) {};
\end{tikzpicture}
\caption{An instance of \ProxyVoting{}.}
\vspace{0.1in}
\label{subfig:Example_UpperBound_Restricted_Instance}
\end{subfigure}
\begin{subfigure}{\textwidth}
\centering
\begin{tikzpicture}[line width=0.8pt]
	\tikzset{candidate/.style = {draw,shape=rectangle,scale=1.1}}
	\tikzset{proxy/.style = {circle, fill=red, minimum size=3pt, inner sep=0pt, outer sep=0pt}}
	\tikzmath{\proxyYoffset=0;\proxyXoffset=0.2;}
		\draw[help lines, dashed, step = 0.5cm] (-6, -0.45) grid (4, 0.45);
		\node[candidate,fill=blue!10,label=above:{$c_1$}] at (-5.5,0){};
		\node[candidate,label=above:{$c_2$}] at (-3.5,0) {};
		\node[candidate,fill=blue!10,label=above:{$c_3$}] at (-3,0) {};
		\node[candidate,fill=blue!10,label=above:{$c_4$}] at (1,0) {};
		\node[candidate,fill=blue!10,label=above:{$c_5$}] at (3.5,0) {};
	    \draw[line width=1.2pt, red, dashed] (-3.25, -0.4) -- (-3.25, 0.4);
	    \draw[line width=1.2pt, red, dashed] (-1, -0.4) -- (-1, 0.4);
	    \draw[line width=1.2pt, red, dashed] (2.25, -0.4) -- (2.25, 0.4);
		\node[proxy] at (-3.5,\proxyYoffset) {};
		\node[proxy] at (-3,\proxyYoffset) {};
		\node[proxy] at (1,\proxyYoffset) {};
		\node[proxy] at (3.5,\proxyYoffset) {};
		\node at (4.375,0) {};
\end{tikzpicture}
\caption{The proxy bisectors computed by the algorithm are shown as dashed red lines, and the reference candidates are highlighted in blue. The locations of proxies are shown as solid red circles.}
\label{subfig:Example_UpperBound_Restricted_Proxy_Bisectors}
\end{subfigure}
\caption{Illustrating the execution of the algorithm in \Cref{thm:UpperBound_Restricted} on a toy example.}
\label{fig:Example_UpperBound_Restricted}
\end{figure}

To compute the desired proxy arrangement, our algorithm (see Algorithm~\ref{alg:UpperBound_Restricted} in the appendix) computes a set of \emph{proxy bisectors} between adjacent proxies, according to the following strategy. Starting from the leftmost candidate as the \emph{reference}, the bisector between the furthest $\theta$-close and the closest $\theta$-far candidates from the reference on the right is chosen as the first proxy bisector. (In \Cref{subfig:Example_UpperBound_Restricted_Proxy_Bisectors}, this is the bisector between $c_2$ and $c_3$.) 
In the next iteration, the candidate immediately to the right of the previous proxy bisector is chosen as the new reference in order to compute the next proxy bisector. (Thus, when $c_3$ is the reference in the next iteration, the furthest $\theta$-close candidate from $c_3$ on the right is $c_3$ itself, and $c_4$ is the closest $\theta$-far candidate. Therefore, the proxy bisector is chosen as the candidate bisector between $c_3$ and $c_4$.) This process repeats until the rightmost candidate $c_m$ either becomes a reference or is $\theta$-close to one.

Since, by construction, all proxy bisectors coincide with candidate bisectors, the desired proxy arrangement is realized by placing proxies on the equidistant candidates next to each proxy bisector. It is easy to show that the number of proxy bisectors is at most $\lfloor\frac{1}{\theta}\rfloor$, and therefore the number of proxies is at most $2\lfloor\frac{1}{\theta}\rfloor$. The $\theta$-representation of this proxy arrangement follows from the fact that all candidates between consecutive proxy bisectors are $\theta$-close.
\end{proofsketch}

\Cref{prop:LowerBound_Restricted} shows that the upper bound derived in \Cref{thm:UpperBound_Restricted} is tight.

\begin{restatable}[\textbf{Lower bound under restricted positioning}]{prop}{LowerBoundRestricted}
Given any $\theta \in (0,1)$, there exists an instance for which any $\theta$-representative proxy arrangement under restricted positioning requires at least $2(\frac{1}{\theta}-1)$ proxies if $\frac{1}{\theta} \in \mathbb{N}$, and $2\lfloor\frac{1}{\theta}\rfloor$ proxies otherwise.
\label{prop:LowerBound_Restricted}
\end{restatable}

\begin{proof}
Let $p \in \mathbb{N}$ denote the unique positive integer such that $\frac{1}{p} \leq \theta < \frac{1}{p-1}$. Observe that $p-1$ equals $\frac{1}{\theta} - 1$ when $\frac{1}{\theta} \in \mathbb{N}$, and equals $\lfloor\frac{1}{\theta}\rfloor$ otherwise. Therefore, it suffices to show that any $\theta$-representative arrangement requires $2p-2$ proxies.

Our construction of the lower bound instance will depend on whether $p$ is even or odd. Specifically, let $\eps \coloneqq \frac{1-(p-1)\theta}{5(p-1)/2}$ (if $p$ is odd) or $\eps \coloneqq \frac{1-(p-1)\theta}{(5p/2-3)}$ (if $p$ is even), and observe that $\eps > 0$ in both cases. The distinction between even and odd cases is made in order to ensure that the distance between the extreme candidates is equal to $1$.

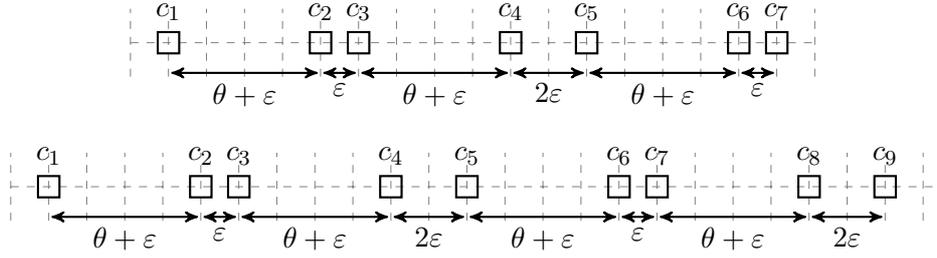
\begin{figure}[h]
\centering
\begin{subfigure}{\textwidth}
\centering
\begin{tikzpicture}[line width=0.8pt]
	\tikzset{candidate/.style = {draw,shape=rectangle,scale=1.1}}
		\draw[help lines, dashed, step = 0.5cm] (-6, -0.45) grid (3, 0.45);
		\node[candidate,label=above:{$c_1$}] at (-5.5,0){};
		\node[candidate,label=above:{$c_2$}] at (-3.5,0) {};
		\node[candidate,label=above:{$c_3$}] at (-3,0) {};
		\node[candidate,label=above:{$c_4$}] at (-1,0) {};
		\node[candidate,label=above:{$c_5$}] at (0,0) {};
		\node[candidate,label=above:{$c_6$}] at (2,0) {};
		\node[candidate,label=above:{$c_7$}] at (2.5,0) {};
	    \draw[>=stealth', <->, shorten <=1pt, shorten >=1pt] (-5.5,-0.4) -- node[below] {$\theta+\eps$} ++ (2,0);
		\draw[>=stealth', <->, shorten <=1pt, shorten >=1pt] (-3.5,-0.4) -- node[below] {$\eps$} ++ (0.5,0);
	    \draw[>=stealth', <->, shorten <=1pt, shorten >=1pt] (-3,-0.4) -- node[below] {$\theta+\eps$} ++ (2,0);
	    \draw[>=stealth', <->, shorten <=1pt, shorten >=1pt] (-1,-0.4) -- node[below] {$2\eps$} ++ (1,0);
	    \draw[>=stealth', <->, shorten <=1pt, shorten >=1pt] (0,-0.4) -- node[below] {$\theta+\eps$} ++ (2,0);
	    \draw[>=stealth', <->, shorten <=1pt, shorten >=1pt] (2,-0.4) -- node[below] {$\eps$} ++ (0.5,0);
\end{tikzpicture}
\vspace{0.1in}
\end{subfigure}
\begin{subfigure}{\textwidth}
\centering
\begin{tikzpicture}[line width=0.8pt]
	\tikzset{candidate/.style = {draw,shape=rectangle,scale=1.1}}
		\draw[help lines, dashed, step = 0.5cm] (-6, -0.45) grid (6.25, 0.45);
		\node[candidate,label=above:{$c_1$}] at (-5.5,0){};
		\node[candidate,label=above:{$c_2$}] at (-3.5,0) {};
		\node[candidate,label=above:{$c_3$}] at (-3,0) {};
		\node[candidate,label=above:{$c_4$}] at (-1,0) {};
		\node[candidate,label=above:{$c_5$}] at (0,0) {};
		\node[candidate,label=above:{$c_6$}] at (2,0) {};
		\node[candidate,label=above:{$c_7$}] at (2.5,0) {};
		\node[candidate,label=above:{$c_8$}] at (4.5,0) {};
		\node[candidate,label=above:{$c_9$}] at (5.5,0) {};
	    \draw[>=stealth', <->, shorten <=1pt, shorten >=1pt] (-5.5,-0.4) -- node[below] {$\theta+\eps$} ++ (2,0);
		\draw[>=stealth', <->, shorten <=1pt, shorten >=1pt] (-3.5,-0.4) -- node[below] {$\eps$} ++ (0.5,0);
	    \draw[>=stealth', <->, shorten <=1pt, shorten >=1pt] (-3,-0.4) -- node[below] {$\theta+\eps$} ++ (2,0);
	    \draw[>=stealth', <->, shorten <=1pt, shorten >=1pt] (-1,-0.4) -- node[below] {$2\eps$} ++ (1,0);
	    \draw[>=stealth', <->, shorten <=1pt, shorten >=1pt] (0,-0.4) -- node[below] {$\theta+\eps$} ++ (2,0);
	    \draw[>=stealth', <->, shorten <=1pt, shorten >=1pt] (2,-0.4) -- node[below] {$\eps$} ++ (0.5,0);
	    \draw[>=stealth', <->, shorten <=1pt, shorten >=1pt] (2.5,-0.4) -- node[below] {$\theta+\eps$} ++ (2,0);
	    \draw[>=stealth', <->, shorten <=1pt, shorten >=1pt] (4.5,-0.4) -- node[below] {$2\eps$} ++ (1,0);
\end{tikzpicture}
\end{subfigure}
\caption{Lower bound instance in \Cref{prop:LowerBound_Restricted} when $\frac{1}{4} \leq \theta < \frac{1}{3}$ (top) and $\frac{1}{5} \leq \theta < \frac{1}{4}$ (bottom).}
\label{fig:LowerBound_Restricted}
\end{figure}

The instance consists of the candidate $c_1$ on the left, followed by a sequence of pairs of candidates $\{c_{2i},c_{2i+1}\}$ for $i \in \mathbb{N}$ such that each pair is at a distance of $\theta+\eps$ from the preceding pair, and the gap between the candidates in each pair alternates between $\eps$ and $2\eps$ (see \Cref{fig:LowerBound_Restricted} for an example). In other words, the sequence of distances between adjacent pair of candidates from left to right is given by $\theta+\eps,\eps,\theta+\eps,2\eps,\theta+\eps,\eps$ and so on. The rightmost candidate is $c_{2p-1}$. It is easy to see that all candidate locations are rational, and that the distance between the extreme candidates is equal to $1$.

The fact that the candidate pairs $\{c_1,c_2\}$, $\{c_3,c_4\}$, $\{c_5,c_6\}$, and so on are each $\theta$-far necessitates that for every candidate bisector between these pairs, there must exist a proxy bisector between adjacent proxies that coincides with it~(contrapositive of \Cref{lem:Bisectors_Non_Coincident}). Furthermore, due to the alternating gaps property, any $\theta$-representative proxy arrangement is required to place proxies on every candidate except for the last candidate $c_{2p-1}$ (see \Cref{subsubsec:Proof_LowerBound_Restricted} 
in the appendix for a detailed argument). This implies that any $\theta$-representative proxy arrangement for the above instance requires at least $2p-2$ proxies.
\end{proof}

\section{Unrestricted Positioning of Proxies}
\label{sec:Unrestricted_Proxies}

Let us now turn our attention to the \emph{unrestricted} setting wherein the proxies can be placed anywhere on the real line. Clearly, any feasible proxy arrangement in the restricted model is also feasible under the unrestricted model. Therefore, the optimal number of proxies under the latter setting is at most that under the former; in fact, \Cref{eg:Unrestricted_Beats_Restricted} shows that the separation between the two models can be strict.

\begin{example}[\textbf{Unrestricted positioning uses fewer proxies than restricted}]

Fix $\theta = \frac{1}{3}$ and let $\eps > 0$ be sufficiently small. Consider the instance shown in \Cref{fig:Unrestricted_Beats_Restricted} consisting of four candidates $c_1=0$, $c_2=\frac{1}{3}+\eps$, $c_3=\frac{2}{3}-\eps$, and $c_4=1$. Notice that the adjacent candidate pairs $\{c_1,c_2\}$ and $\{c_3,c_4\}$ are $\theta$-far. Thus, by the contrapositive of \Cref{lem:Bisectors_Non_Coincident}, any $\theta$-representative proxy arrangement must have bisectors between adjacent proxies that coincide with the candidate bisectors between these pairs.

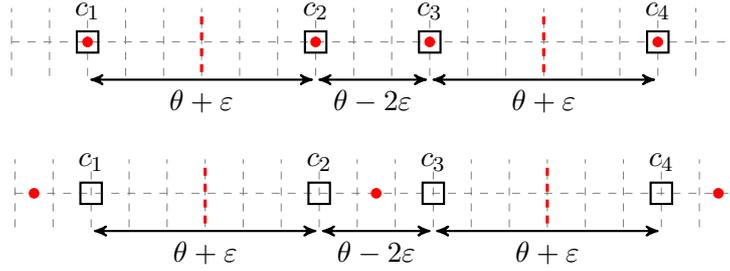
\begin{figure}[h]
\centering
\begin{subfigure}{\textwidth}
\centering
\begin{tikzpicture}[line width=0.8pt]
	\tikzset{candidate/.style = {draw,shape=rectangle,scale=1.1}}
	\tikzset{Proxy/.style = {circle, fill=red, minimum size=4pt, inner sep=0pt, outer sep=0pt}}
	\tikzmath{\proxyYoffset=0;}
		\draw[help lines, dashed, step = 0.5cm] (-5, -0.45) grid (4.5, 0.45);
		\node[candidate,label=above:{$c_1$}] at (-4,0){};
		\node[candidate,label=above:{$c_2$}] at (-1,0) {};
		\node[candidate,label=above:{$c_3$}] at (0.5,0) {};
		\node[candidate,label=above:{$c_4$}] at (3.5,0) {};
	    \draw[line width=1.2pt, red, dashed] (-2.5, -0.4) -- (-2.5, 0.4);
	    \draw[line width=1.2pt, red, dashed] (2, -0.4) -- (2, 0.4);
	\node[Proxy] at (-4,\proxyYoffset) {};
	\node[Proxy] at (-1,\proxyYoffset) {};
	\node[Proxy] at (0.5,\proxyYoffset) {};
	\node[Proxy] at (3.5,\proxyYoffset) {};
	    \draw[>=stealth', <->, shorten <=1pt, shorten >=1pt] (-4,-0.5) -- node[below] {$\theta+\eps$} ++ (3,0);
	    \draw[>=stealth', <->, shorten <=1pt, shorten >=1pt] (-1,-0.5) -- node[below] {$\theta-2\eps$} ++ (1.5,0);
	    \draw[>=stealth', <->, shorten <=1pt, shorten >=1pt] (0.5,-0.5) -- node[below] {$\theta+\eps$} ++ (3,0);
\end{tikzpicture}
\vspace{0.1in}
\end{subfigure}
\begin{subfigure}{\textwidth}
\centering
\begin{tikzpicture}[line width=0.8pt]
	\tikzset{candidate/.style = {draw,shape=rectangle,scale=1.1}}
	\tikzset{Proxy/.style = {circle, fill=red, minimum size=4pt, inner sep=0pt, outer sep=0pt}}
	\tikzmath{\proxyYoffset=0;}
		\draw[help lines, dashed, step = 0.5cm] (-5, -0.45) grid (4.5, 0.45);
		\node[candidate,label=above:{$c_1$}] at (-4,0){};
		\node[candidate,label=above:{$c_2$}] at (-1,0) {};
		\node[candidate,label=above:{$c_3$}] at (0.5,0) {};
		\node[candidate,label=above:{$c_4$}] at (3.5,0) {};
	    \draw[line width=1.2pt, red, dashed] (-2.5, -0.4) -- (-2.5, 0.4);
	    \draw[line width=1.2pt, red, dashed] (2, -0.4) -- (2, 0.4);
	\node[Proxy] at (-4.75,\proxyYoffset) {};
	\node[Proxy] at (-0.25,\proxyYoffset) {};
	\node[Proxy] at (4.25,\proxyYoffset) {};
	    \draw[>=stealth', <->, shorten <=1pt, shorten >=1pt] (-4,-0.5) -- node[below] {$\theta+\eps$} ++ (3,0);
	    \draw[>=stealth', <->, shorten <=1pt, shorten >=1pt] (-1,-0.5) -- node[below] {$\theta-2\eps$} ++ (1.5,0);
	    \draw[>=stealth', <->, shorten <=1pt, shorten >=1pt] (0.5,-0.5) -- node[below] {$\theta+\eps$} ++ (3,0);
\end{tikzpicture}
\end{subfigure}
\caption{Unrestricted positioning (bottom) requires strictly fewer proxies than restricted positioning (top). The proxies in each case are shown as solid red circles.}
\label{fig:Unrestricted_Beats_Restricted}
\end{figure}

Under restricted positioning, a $\theta$-representative proxy arrangement must place proxies on all candidates, using four proxies in total (see top figure in \Cref{fig:Unrestricted_Beats_Restricted}). By contrast, under unrestricted positioning, there exists a feasible proxy arrangement that only uses three proxies, namely $p_1 = -\frac{1}{6}+\eps$, $p_2 = 0.5$, and $p_3 = \frac{7}{6}+\eps$ (see bottom figure in \Cref{fig:Unrestricted_Beats_Restricted}).

Notice that the saving in the number of proxies in the unrestricted case was achieved by ``merging'' the second and the third proxies. This, in turn, forces the first and the fourth proxies to fall outside of $[0,1]$, thus highlighting the importance of allowing the proxies to be anywhere on the real line.
\label{eg:Unrestricted_Beats_Restricted}
\end{example}

\subsection{Algorithm for Computing Optimal Number of Unrestricted Proxies}

Let us now state our main result for unrestricted proxies.

\begin{restatable}[\textbf{Optimal proxy arrangement under unrestricted positioning}]{theorem}{UnrestrictedProxiesOptimal}
There is a polynomial-time algorithm that, given any instance $\langle \C, \theta \rangle$ of \ProxyVoting{} as input, terminates in polynomial time and returns an optimal $\theta$-representative proxy arrangement.
\label{thm:Unrestricted_Proxies_Optimal}
\end{restatable}

The proof of \Cref{thm:Unrestricted_Proxies_Optimal} is technically the most involved part of the paper and is presented in \Cref{subsubsec:Proof_Unrestricted_Proxies_Optimal} in the appendix. Here we will outline a brief sketch of the proof.

\begin{proofsketch}
For any fixed $k \in [m]$, our algorithm decides whether there exists a feasible (i.e., $\theta$-representative) arrangement of $k$ proxies for the given instance. The smallest $k$ with a positive answer is returned as the output.

Given a proxy arrangement $\P = \{p_1,\dots,p_k\}$ (where $p_1 < p_2 < \dots < p_k$), let us define the Voronoi cell $W_j$ of proxy $p_j$ as the set of all locations of voters whose closest proxy is $p_j$, i.e., $W_j \coloneqq \{v \in [0,1]: |v - p_j| \leq |v - p_\ell| \text{ for any } \ell \neq j\}$, where ties are broken according to a directionally-consistent tie-breaking rule (\Cref{defn:Tie-Breaking}). Notice that $v \in W_j$ if and only if $p_j = p^v$.

At a high level, the algorithm uses dynamic programming to compute the set of feasible locations (or the \emph{feasibility set}) of the proxy $p_{j+1}$ using the feasibility set of the preceding proxy $p_j$. It maintains the property that for each point $s$ in the feasibility set of proxy $p_{j+1}$, there exists some point $s' < s$ in the feasibility set of proxy $p_j$ such that $\theta$-representation is satisfied for all voters in $W_j$ (see \Cref{lem:Unrestricted_Proxies_Characterization} in the appendix for a formal characterization result). Such a pair of locations $\{s',s\}$ is said to be \emph{mutually feasible}.

To begin with, the feasibility set $F^1$ of the leftmost proxy $p_1$ is initialized as the union of the Voronoi cells of all candidates that are $\theta$-close to the leftmost candidate $c_1$, along with the region $(-\infty,0)$.\footnote{Recall from \Cref{eg:Unrestricted_Beats_Restricted} that proxies can lie outside $[0,1]$ in an optimal arrangement under unrestricted positioning.} This corresponds to the set of all positions of the leftmost proxy such that $\theta$-representation is satisfied for the leftmost voter.

The feasibility set $F^{j+1}$ of the proxy $p_{j+1}$ is computed using $F^j$ as follows: Consider the restriction of $F^j$ to the Voronoi cell of candidate $c_h$ (denoted by $F^{j,h}$), as illustrated by the interval $[x,y]$ in \Cref{fig:Optimal_Unrestricted_Reflection_ProofSketch} (thus, $\top(p_j) = c_h$). We will use $F^{j,h}$ to compute $F^{j+1,i}$, which is the restriction of $F^{j+1}$ to the Voronoi cell of candidate $c_i$.

Let $\myoverbar{b}_h$ denote the candidate bisector between the furthest $\theta$-close and closest $\theta$-far candidates to the right of $c_h$. Similarly, let $\myunderbar{b}_i$ denote the candidate bisector between the furthest $\theta$-close and closest $\theta$-far candidates to the left of $c_i$. 

By our characterization result, it follows that any point $z \in F^{j+1,i}$ that is mutually feasible with some point in $F^{j,h}$ must be (weakly) to the left of the mirror image of the point $x$ about the bisector $\myoverbar{b}_h$ (namely, $x'$). This is because the proxy bisector between $p_{j}$ and $p_{j+1}$ has to be (weakly) to the left of $\myoverbar{b}_h$. Similarly, $z$ must be (weakly) to the right of the mirror image of point $y$ about the bisector $\myunderbar{b}_i$ (namely, $y'$). The contribution of the set $[x,y]$ to $F^{j+1,i}$ is given by $V_i \cap [y',x']$. In general, the restriction $F^{j,h}$ could comprise of several disjoint intervals. In \Cref{subsubsec:Proof_Unrestricted_Proxies_Optimal} in the appendix, we describe how the feasibility set $F^{j+1}$ can nevertheless be efficiently computed.

\begin{figure}[h]
\centering
%
\begin{tikzpicture}[line width=0.8pt]
	\tikzset{candidate/.style = {draw,shape=rectangle,scale=1.1}}
	\tikzset{proxy/.style = {circle, draw, red, minimum size=3pt, inner sep=0pt, outer sep=0pt}}
	\tikzmath{\proxyYoffset=-0.3;\proxyXoffset=0.2;}
	    \fill [gray!20] (-7,-0.3) rectangle ++(2.5,0.6);
	    \fill [gray!20] (-1.5,-0.3) rectangle ++(4,0.6);
		\draw[help lines, dashed, step = 0.5cm] (-7.5, -0.45) grid (3, 0.45);
		\node[candidate,fill=blue!20,label=above:{$c_h$}] at (-5,0){};
		\node[candidate] at (-4,0) {};
		\node[candidate] at (-3,0) {};
		\node[candidate] at (-2,0) {};
		\node[candidate,fill=blue!20,label=above:{$c_i$}] at (-1,0) {};
		%
	    \draw[line width=1.2pt, dashed] (-3.5, -0.4) -- (-3.5, 0.4);
	    \node[label=above:{$\myunderbar{b}_i$}] at (-3.5,0.3) {};
	    \draw[line width=1.2pt, dashed] (-2.5, -0.4) -- (-2.5, 0.4);
	    \node[label=above:{$\myoverbar{b}_h$}] at (-2.5,0.3) {};
	    \node[circle,minimum size=1.5pt,inner sep=1pt,fill=red,label=below:{$x$}] (test) at (-7,0) {};
	    \node[circle,minimum size=1.5pt,inner sep=1pt,fill=red,label=below:{$y$}] (test) at (-6.5,0) {};
	    \fill [red] (-6.99,0.02) rectangle (-6.49,-0.02);
	    \node[circle,minimum size=1.5pt,inner sep=1pt,fill=red,label=below:{$y'$}] (test) at (-0.5,0) {};
	    \node[circle,minimum size=1.5pt,inner sep=1pt,fill=red,label=below:{$x'$}] (test) at (2,0) {};
	    \fill [red] (-0.49,0.02) rectangle (1.99,-0.02);
	    %
	    \dimline[extension start length=0,extension end length=0]{(-5,-0.5)}{(-3,-0.5)}{$\theta$};
	    %
	    \dimline[extension start length=0,extension end length=0]{(-3,-0.5)}{(-1,-0.5)}{$\theta$};
\end{tikzpicture}
\caption{}
\label{fig:Optimal_Unrestricted_Reflection_ProofSketch}
\end{figure}

Having computed the feasibility set $F^k$ of the rightmost proxy, the algorithm now checks whether it overlaps with the union of the Voronoi cells of candidates that are $\theta$-close to the rightmost candidate $c_m$ along with the region $(1,\infty)$. If yes, then there exists a feasible location of the proxy $p_k$ that is $\theta$-representative for the rightmost voter. By the aforementioned property, there must exist a feasible location $p_{k-1}$ of the previous proxy such that $p_{k-1}$ and $p_k$ are mutually feasible, thus implying $\theta$-representation for the voters in $W_{k-1}$. Similarly, there must exist a feasible location of the proxy $p_{k-2}$ that is mutually feasible with $p_{k-1}$, and so on. Continuing backwards in this manner, we obtain a set of proxy locations $p_1,\dots,p_k$ wherein the adjacent pairs are mutually feasible, which immediately implies $\theta$-representation. On the other hand, the absence of an overlap certifies that with $k$ proxies, there is no $\theta$-representative proxy arrangement in the given instance.
\end{proofsketch}

\subsection{Upper and Lower Bounds for Unrestricted Positioning of Proxies}

To compute an upper bound on the number of proxies, we will show that an algorithm similar to that for the restricted setting turns out to be useful. We will defer the detailed description of our algorithm and its formal analysis to \Cref{subsubsec:Proof_UpperBound_Unrestricted} in the appendix, 
and instead revisit the example from \Cref{fig:Example_UpperBound_Restricted} considered previously in the restricted case.

Our algorithm proceeds in two phases. The first phase is identical to that of the algorithm for the restricted case (\Cref{thm:UpperBound_Restricted}), and returns a set of proxy bisectors. In the second phase, the algorithm starts with a proxy arrangement that is consistent with the proxy bisectors computed in Phase 1 by placing a pair of equidistant proxies on either side of each bisector. This results in twice as many proxies as there are bisectors. To shrink this number down to $3/2$ times the number of bisectors, the algorithm utilizes the additional flexibility of the unrestricted setting via an ``expand and merge'' step. Specifically, the algorithm pulls all the proxies away from their bisectors at equal speeds until there is a ``collision'' event in some interval (recall that an interval is the area between adjacent proxy bisectors). At this point, the two proxies in that interval can be merged into a single proxy, and the locations of their `partner' proxies are frozen (see \Cref{fig:Example_UpperBound_Unrestricted}). This process is repeated until all proxies are frozen.

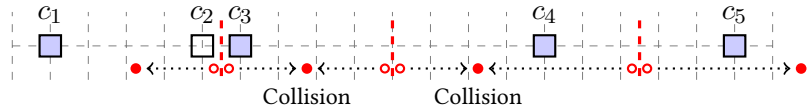
\begin{figure}[h]
\centering
\begin{tikzpicture}[line width=0.8pt]
	\tikzset{candidate/.style = {draw,shape=rectangle,scale=1.1}}
	\tikzset{Oldproxy/.style = {circle, draw,red, minimum size=3pt, inner sep=0pt, outer sep=0pt}}
	\tikzset{Newproxy/.style = {circle, fill=red, minimum size=4pt, inner sep=0pt, outer sep=0pt}}
	\tikzmath{\proxyYoffset=-0.3;\proxyXoffset=0.2;}
		\draw[help lines, dashed, step = 0.5cm] (-6, -0.45) grid (4, 0.45);
		\node[candidate,fill=blue!20,label=above:{$c_1$}] at (-5.5,0){};
		\node[candidate,label=above:{$c_2$}] at (-3.5,0) {};
		\node[candidate,fill=blue!20,label=above:{$c_3$}] at (-3,0) {};
		\node[candidate,fill=blue!20,label=above:{$c_4$}] at (1,0) {};
		\node[candidate,fill=blue!20,label=above:{$c_5$}] at (3.5,0) {};
	    \draw[line width=1.2pt, red, dashed] (-3.25, -0.4) -- (-3.25, 0.4);
	    \draw[line width=1.2pt, red, dashed] (-1, -0.4) -- (-1, 0.4);
	    \draw[line width=1.2pt, red, dashed] (2.25, -0.4) -- (2.25, 0.4);
		\node[Oldproxy] at (-3.25-0.5*\proxyXoffset,\proxyYoffset) {};
		\node[Oldproxy] at (-3.25+0.5*\proxyXoffset,\proxyYoffset) {};
		\node[Oldproxy] at (-1-0.5*\proxyXoffset,\proxyYoffset) {};
		\node[Oldproxy] at (-1+0.5*\proxyXoffset,\proxyYoffset) {};
		\node[Oldproxy] at (2.25-0.5*\proxyXoffset,\proxyYoffset) {};
		\node[Oldproxy] at (2.25+0.5*\proxyXoffset,\proxyYoffset) {};
		\node[Newproxy] at (-4.375,\proxyYoffset) {};
		\node[Newproxy] (a) at (-2.125,\proxyYoffset) {};
		\node[below=1pt of a] {\footnotesize{Collision}};
		\node[Newproxy] (a) at (0.125,\proxyYoffset) {};
		\node[below=1pt of a] {\footnotesize{Collision}};
		\node[Newproxy] at (4.375,\proxyYoffset) {};
		\draw[->, dotted, shorten <=4pt, shorten >=4pt] (-3.25-0.5*\proxyXoffset,\proxyYoffset) -- (-4.375,\proxyYoffset);
		\draw[->, dotted, shorten <=4pt, shorten >=4pt] (-3.25+0.5*\proxyXoffset,\proxyYoffset) -- (-2.125,\proxyYoffset);
		\draw[->, dotted, shorten <=4pt, shorten >=4pt] (-1-0.5*\proxyXoffset,\proxyYoffset) -- (-2.125,\proxyYoffset);
		\draw[->, dotted, shorten <=4pt, shorten >=4pt] (-1+0.5*\proxyXoffset,\proxyYoffset) -- (0.125,\proxyYoffset);
		\draw[->, dotted, shorten <=4pt, shorten >=4pt] (2.25-0.5*\proxyXoffset,\proxyYoffset) -- (0.125,\proxyYoffset);
		\draw[->, dotted, shorten <=4pt, shorten >=4pt] (2.25+0.5*\proxyXoffset,\proxyYoffset) -- (4.375,\proxyYoffset);
\end{tikzpicture}
%
\caption{Illustrating the execution of the algorithm in \Cref{thm:UpperBound_Unrestricted} on a toy example. The initial and the final proxy locations are shown as empty and solid red circles, respectively.}
\label{fig:Example_UpperBound_Unrestricted}
\end{figure}

By construction, all candidates within each interval are $\theta$-close. Furthermore, since each voter is in the same interval as its closest proxy, $\theta$-representation is satisfied for all voters. By an analysis similar to that in \Cref{thm:UpperBound_Restricted}, it follows that the number of proxy bisectors is at most $\lceil\frac{1}{\theta}\rceil$. The expand-and-merge step ensures that no two consecutive intervals can have two proxies each. This readily implies the desired bound of $\frac{3}{2}\lceil\frac{1}{\theta}\rceil$ for the number of proxies.

\begin{restatable}[\textbf{Upper bound under unrestricted positioning}]{theorem}{UpperBoundUnrestricted}
Given any instance $\langle \C, \theta \rangle$ of \ProxyVoting{} as input, there exists a $\theta$-representative arrangement consisting of at most $\frac{3}{2}\lceil\frac{1}{\theta}\rceil$ proxies. Furthermore, such an arrangement can be computed in polynomial time.
\label{thm:UpperBound_Unrestricted}
\end{restatable}

Notice that \Cref{thm:UpperBound_Unrestricted} can be used to obtain upper bounds for the dual problem to \ProxyVoting{} wherein the input consists of a proxy budget $k$ and the goal is to provide an upper bound on $\theta$. Indeed, given as input a budget of $k > 1$ proxies, we can invoke our algorithm with $\theta = 1/\lfloor\frac{2k}{3}\rfloor$. Then, by \Cref{thm:UpperBound_Unrestricted}, we know that the proxy arrangement returned by the algorithm is $\left(1/\lfloor\frac{2k}{3}\rfloor\right)$-representative and uses at most $\frac{3}{2}\lceil\frac{1}{\theta}\rceil \leq k$ proxies, which is within the given budget. \Cref{cor:UpperBound_Unrestricted_Theta} formalizes this observation.

\begin{restatable}[\textbf{Upper bound on $\theta$}]{corollary}{UpperBoundUnrestrictedTheta}
There is a polynomial-time algorithm that, given any positive integer $k \in \mathbb{N}$ as input, returns an arrangement of at most $k$ proxies that is $\left(1/\lfloor\frac{2k}{3}\rfloor\right)$-representative.
\label{cor:UpperBound_Unrestricted_Theta}
\end{restatable}

Next, we will show that the dependence on $\theta$ in the upper bound cannot be improved.

\begin{restatable}[\textbf{Lower bound under unrestricted positioning}]{prop}{LowerBoundForK}
Given any $\theta \in (0,1)$, there exists an instance such that any $\theta$-representative proxy arrangement requires at least $\lceil \frac{1}{\theta} \rceil$ proxies.
\label{prop:LowerBound_Unrestricted}
\end{restatable}
\begin{proof}
%
Define $\theta' \in (0,1)$ as follows:
\[
    \theta' \coloneqq \begin{cases}
        \frac{1}{p-1}, & \text{if } \theta = \frac{1}{p} \text{ for some positive integer } p\\
        \frac{1}{p}, & \text{if } \frac{1}{p+1} < \theta < \frac{1}{p} \text{ for some positive integer } p.
        \end{cases}
\]

Observe that $\theta' > \theta$. Consider an instance where the candidates are evenly spaced at a distance of $\theta'$, as shown in \Cref{fig:LowerBound_Unrestricted}. Notice that the total number of candidates is $m = \frac{1}{\theta'} + 1 = \lceil \frac{1}{\theta} \rceil$.

\begin{figure}[h]
\centering
\begin{tikzpicture}[line width=0.8pt]
	\tikzset{candidate/.style = {draw,shape=rectangle,scale=1.1}}
	\tikzset{voter/.style = {cross out, draw=red, minimum size=4pt, inner sep=0pt, outer sep=0pt},voter/.default = {2pt}}
	\tikzset{
		pics/voter/.style args={#1,#2,#3}{
			code={
			\node[cross out, draw=red, minimum size=4pt, inner sep=0pt, outer sep=0pt] (test1) at (#1,0) {};
			\node[] (test2) at (#1,-1) {#2};
			\node[] (test3) at (#1,-1.35) {#3};
			\draw[->,shorten >=4pt] (test2) to (test1);
     		}
  		}
	}
		\draw[thick, dashed] (-2.25, -0.4) -- (-2.25, 0.4);
		\draw[thick, dashed] (-0.75, -0.4) -- (-0.75, 0.4);
		\draw[thick, dashed] (2.25, -0.4) -- (2.25, 0.4);
		\node[candidate,label=above:{$c_1$}] (1) at (-3,0){};
		\node[candidate,label=above:{$c_2$}] (2) at (-1.5,0) {};
		\node[candidate,label=above:{$c_3$}] (3) at (0,0) {};
		\node[candidate,label=above:{$c_{m-1}$}] (4) at (1.5,0) {};
		\node[candidate,label=above:{$c_m$}] (5) at (3,0) {};
	    \path (3) -- node[auto=false]{\ldots} (4);
	    \dimline[extension start length=0,extension end length=0]{(-3,-0.5)}{(-1.5,-0.5)}{$\theta'$};
	    \dimline[extension start length=0,extension end length=0]{(-1.5,-0.5)}{(0,-0.5)}{$\theta'$};
	    \dimline[extension start length=0,extension end length=0]{(0,-0.5)}{(1.5,-0.5)}{$\theta'$};
\end{tikzpicture}
\caption{Lower bound on the number of proxies under unrestricted positioning (\Cref{prop:LowerBound_Unrestricted})}
\label{fig:LowerBound_Unrestricted}
\end{figure}
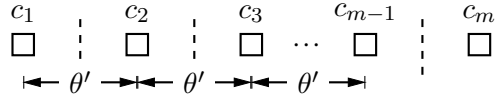

Since any pair of adjacent candidates are $\theta$-far, it must be that each candidate is the favorite candidate of some proxy. Indeed, in any arrangement with fewer than $m$ proxies, there must exist a candidate, say $c$, that is not the favorite candidate of any proxy. Consider a voter $v$ that is located at $c$. The distance between the favorite candidate of $v$ and that of $v$'s nearest proxy must then be strictly greater than $\theta$, which violates $\theta$-representation. Thus, $m$ proxies are necessary, which gives the desired bound.
\end{proof}

\section{Representative Election Outcomes}

So far, we have focused on achieving $\theta$-representation for each individual voter (i.e., in the ``input'' space). A natural question is whether a similar guarantee can be achieved in the ``outcome'' space as well. Specifically, when is it that the outcome under direct voting (where the preference profile consists of voter's preferences over the candidates) is close to the outcome under proxy voting (where each voter's preference is substituted by that of its closest proxy). This is formally defined as follows.

\begin{definition}
Let $P \in [0,1]^n$ denote a preference profile consisting of the votes of the $n$ voters, and let $Q \coloneqq \{p^v:v\in P \}$ be the preference profile derived from $P$ by replacing the vote of each voter $v$ by that of its closest proxy $p^v$ (note that $Q$ can be a multiset). Then, an arrangement of proxies is said to be \emph{$\theta$-representative under a voting rule $r: \mathbb{R}^n \rightarrow \C$} if, for any preference profile $P \in [0,1]^n$, the candidate $r(P)$ is $\theta$-close to the candidate $r(Q)$. Here, $r(P)$ and $r(Q)$ are the outcomes under direct and proxy voting, respectively.
\label{defn:Theta_rule}
\end{definition}

It turns out that $\theta$-representation of a proxy arrangement implies $\theta$-representation under any strict-Condorcet rule with {\em consistent} tie-breaking (i.e., any strict-Condorcet rule which must choose a weak Condorcet winner as the winner, and in case there is more than one weak Condorcet winner, the leftmost (or, the rightmost) one must be chosen). 
\begin{restatable}{prop}{Condorcet}
\label{prop:rep-voting}
Let $r$ be a strict-Condorcet rule with consistent tie-breaking. If an arrangement of proxies is $\theta$-representative, then it is $\theta$-representative under the rule $r$.
\end{restatable}

In light of Proposition~\ref{prop:rep-voting}, all $\theta$-representation results proved in this paper (see Table~\ref{tab:Results}) naturally extend to $\theta$-representation under strict-Condorcet rules. The proof of \Cref{prop:rep-voting} is presented in \Cref{sec:Proof_rep-voting} in the appendix.

\section{Concluding Remarks}

In this paper, we gave efficient algorithms for computing optimal sets of $\theta$-representative proxies, as well as proved upper and lower bounds on the number of proxies needed to achieve this property. Unlike in most related work, our proxy arrangements do not depend on the voter locations, and will remain representative for any set of voters. In fact, all our results hold (although some of the proofs become far more complex) even with additional requirements on the proxy arrangement, such as requiring that for every voter $v\in[0,1]$, its closest proxy is within a distance $\theta$, in addition to them being $\theta$-representative.

Many interesting open problems remain, however, beginning with closing the gap between the upper and lower bounds for the number of proxies under unrestricted positioning, and analysing the ability to form $\theta$-representative proxies in more general metric spaces. More generally, it would be interesting to expand the scope of $\theta$-representation from ``representing the top candidate well'' to ``representing the top-$k$ candidates well'', or to more general fairness properties.

\section*{Acknowledgments}
EA acknowledges support from NSF awards CCF-1527497 and CCF-2006286. RV acknowledges support from ONR\#N00014-171-2621 while he was affiliated with Rensselaer Polytechnic Institute, and is currently supported by project no. RTI4001 of the Department of Atomic Energy, Government of India. Part of this work was done while RV was supported by the Prof. R Narasimhan postdoctoral award. Research done in part while ZF was on research leave at Rensselaer Polytechnic Institute. We thank the anonymous reviewers for their very helpful comments and suggestions.

\bibliographystyle{named}
\bibliography{References}

\clearpage
\begin{center}
	\Large{Appendix}
\end{center}

\section{Proof of Theorem~\ref{thm:Opt_Restricted}}
\label{subsubsec:Proof_Opt_Restricted}

\RestrictedProxiesOptimal*
\begin{proof}
In order to compute the optimal proxy arrangement, it suffices to efficiently determine whether, for any fixed $k \in [m]$, there exists an arrangement of $k$ proxies satisfying the desired properties (and if the answer is YES, then return such an arrangement).

Our algorithm relies on the following structural observation: Consider a proxy arrangement $\P$ consisting of $k$ proxies $p_1,\dots,p_k$ such that the rightmost proxy $p_k$ is placed at the candidate location $c_j$. Then, $\P$ satisfies the desired properties (i.e., $\theta$-representation and restricted positioning) if and only if the following conditions hold:
\begin{enumerate}[1.]
	\item $\P$ satisfies the desired properties (i.e., $\theta$-representation and restricted positioning) with respect to the subinstance $\langle\{c_1,c_2,\dots,c_j\},\theta\rangle$, and 
	\item $c_j \in \C_\theta$, where $\C_\theta \coloneqq \{c_i \in \C : |c_i - c_m| \leq \theta\}$ is the set of candidates that are $\theta$-close to the rightmost candidate $c_m$.
\end{enumerate}
Condition 1 specifies the feasibility requirement for $\P$ with respect to the subinstance on the left side of the proxy $p_k$, while Condition 2 encodes a similar restriction for the right side. Indeed, for an extreme voter located at $v = c_m$, its favorite candidate is $\top(v) = c_m$, and the favorite candidate of its closest proxy is $\top(p^v)=c_j$. Therefore, $\theta$-representation for voter $v$ is equivalent to Condition 2. Furthermore, since there are no proxies between $c_j$ and $c_m$, the aforementioned conditions also imply $\theta$-representation for all intermediate voters $v \in [c_j,c_m]$. 

It is easy to see that Condition 2 can be efficiently checked (recall that all candidate locations and the parameter $\theta$ are rational). Therefore, it suffices to establish that Condition 1 can also be decided in polynomial time. To this end, we will provide an algorithm for determining, for every $j \in [m]$ and every $k \in [m]$, whether there exists a feasible proxy arrangement that uses \emph{exactly} $k$ proxies and under which the rightmost proxy coincides with the candidate $c_j$. 

We will present a dynamic programming algorithm to solve this problem. The algorithm computes a binary table $T$ with $m$ rows and $m$ columns, where the entry $T(j,k)$ equals $1$ if there exists a $\theta$-representative proxy arrangement for the subinstance $\langle\{c_1,c_2,\dots,c_j\},\theta\rangle$ that uses exactly $k$ proxies, say $p_1,\dots,p_k$, such that (a) the proxy locations are a subset of the candidate locations (i.e., $\{p_1,\dots,p_k\} \subseteq \{c_1,c_2,\dots,c_j\}$), and (b) the rightmost proxy $p_k$ is placed on the rightmost candidate $c_j$ for the subinstance.

To compute the table $T$, the algorithm starts by setting the diagonal entries $T(j,j) = 1$ for all $j \in [m]$ since $j$ proxies suffice for a subinstance with $j$ candidates. In addition, the upper-triangular entries are set to $0$ (i.e, $T(j,k) = 0$ whenever $k > j$) due to insufficient candidate locations for placing $k$ distinct proxies. When $k < j$, we set $T(j,k) = 1$ if the following two conditions hold:
\begin{enumerate}[(a)]
	\item there exists $i < j$ such that $T(i,k-1) = 1$, and
	\item for all $v \in [c_i,c_j]$, $|\top(v) - \top(p^v)| \leq \theta$;
\end{enumerate}
otherwise we set $T(j,k)= 0$.

Condition (a) ensures that there is a subinstance $\langle c_1,\dots,c_i \rangle$ that admits a feasible proxy arrangement with $k-1$ proxies such that the proxy $p_{k-1}$ is placed at $c_i$,
 while condition (b) ensures $\theta$-representation for the voters that lie between $c_i$ and $c_j$. 

It is easy to see that condition (a) can be efficiently checked. To decide condition (b), the algorithm checks whether the proxy bisector between $p_{k-1}$ and $p_k$ coincides with a bisector between some pair of adjacent candidates.

\begin{enumerate}[(i)]
	\item If yes, then let $c_\ell$ and $c_{\ell+1}$ denote the adjacent candidates such that $\frac{c_\ell+c_{\ell+1}}{2} = \frac{p_k+p_{k-1}}{2}$. Then, the algorithm returns YES if $|c_\ell - c_i| \leq \theta$ and $|c_{\ell+1} - c_j| \leq \theta$, otherwise it returns NO.

	\item Otherwise, the algorithm checks whether the favorite candidate of the voters located on either side of the proxy bisector is $\theta$-close to both $c_i$ and $c_j$. Specifically, let $\delta$ denote the smallest positive gap between any two candidate bisectors between not-necessarily-adjacent candidates (notice that such a $\delta > 0$ must exist), and let $\eps = \delta/3$. Then, the algorithm considers the voters $v^L \coloneqq \frac{p_{k-1} + p_k}{2} - \eps$ and $v^R \coloneqq \frac{p_{k-1} + p_k}{2} + \eps$, and checks whether $|\top(v^L) - \top(p^{v^L})| \leq \theta$ and $|\top(v^R) - \top(p^{v^R})| \leq \theta$, where $p^{v^L}$ and $p^{v^R}$ denote the proxies closest to $v^L$ and $v^R$, respectively. If both of these conditions are true, then the algorithm decides YES for condition (b), otherwise it decides NO.
\end{enumerate}

The reasoning behind the step (i) is the following: When $p_{k-1}=c_i$ and $p_k=c_j$, the proxy bisector between $p_{k-1}$ and $p_k$ coincides with the candidate bisector between $c_i$ and $c_j$ (note that $c_i$ and $c_j$ may not be adjacent). Then, for any voter $v \in [c_i,c_j]$, its closest proxy is either $p_{k-1}$ or $p_k$ (i.e., $p^v \in \{p_{k-1},p_k\}$), and therefore $\top(p^v) \in \{c_i,c_j\}$. Furthermore, under a directionally-consistent tie-breaking rule, voter $v$'s favorite candidate $\top(v)$ must be on the same side of the proxy bisector as $\top(p^v)$. Thus, in order to check for $\theta$-representation, we only need to check whether the candidates closest to the bisector (equivalently, the candidates that are furthest away from $c_i$ and $c_j$) are $\theta$-close to either $c_i$ or $c_j$.

The reasoning behind step (ii) is as follows: Since the proxy bisector between $p_{k-1}$ and $p_k$ does not coincide with any candidate bisector between adjacent candidates, there must exist a candidate, say $c_\ell \in \C$, whose Voronoi cell $V_\ell$ includes the proxy bisector in its relative interior. Then, for any sufficiently small $\eps > 0$ (for example, as chosen above), there must exist voters that are $\eps$-far from the proxy bisector on either side of it whose favorite candidate is $c_\ell$, but the favorite candidates of their respective closest proxies are $c_i$ and $c_j$. Thus, it suffices to check the $\theta$-representation condition for these voters only. For any other voter $v \in [c_i,c_j] \setminus V_\ell$, its favorite candidate and the favorite candidate of its closest proxy are on the same side of the bisector between $c_i$ and $c_j$, and are therefore $\theta$-close.

Overall, we have that condition (b) can also be efficiently verified. This implies that each lower-triangular entry $T(j,k)$ for $k < j$ can be efficiently computed, implying that the algorithm runs in polynomial time.

Once the table $T$ has been computed, the optimal number of proxies for the given instance $\I$ is the smallest $k$ for which $T(j,k)$ equals $1$ for some candidate $c_j \in \C_\theta$, i.e., $\opt(\I) = \min \{k \in [m] : T(j,k) = 1 \text{ for some } c_j \in \C_\theta\}$. To construct the optimal proxy arrangement $p_1,\dots,p_k$, the algorithm computes the positions of the proxies in the reverse order, i.e., first $p_k$, then $p_{k-1}$, and so on. The strategy is as follows: Pick any $j \in [m]$ such that $T(j,\opt(\I)) = 1$, and place the proxy $p_k$ at $c_j$. Next, pick any $i < j$ such that $T(i,\opt(\I)-1) = 1$; we know that such $i$ exists from condition (a) above. The proxy $p_{k-1}$ is then placed at $c_i$, and so on. Note that the proxy arrangement thus constructed satisfies restricted positioning.
%
%
%
\end{proof}

\section{Proof of Theorem~\ref{thm:UpperBound_Restricted}}
\label{subsubsec:Proof_UpperBound_Restricted}

\UpperBoundRestricted*
\begin{proof}
We will show that Algorithm~\ref{alg:UpperBound_Restricted} computes the desired proxy arrangement.

\begin{algorithm}[h]
\DontPrintSemicolon
\footnotesize
 \linespread{1.1}
\KwIn{A \ProxyVoting{} instance $\I = \langle \C, \theta \rangle$.}
\KwOut{A $\theta$-representative arrangement of proxies.}
\BlankLine
Initialize $c \leftarrow c_1$\Comment*{Start with $c_1$ as the reference}
Initialize $j \leftarrow 1$\Comment*{$j$ is the index of proxy bisector}
\While{$c + \theta < 1$\label{algline:UpperBound_Restricted_Location_Check}}{
\tcp{$c^L$ and $c^R$ are the candidates closest to $(c+\theta)$ on left and right, respectively}
$c^L \leftarrow \arg\min_{c_i \in \C \, : \, c_i \leq c+\theta} |c_i - (c+\theta)|$\label{algline:UpperBound_Restricted_Left_Candidate}\;
$c^R \leftarrow \arg\min_{c_i \in \C \, : \, c_i > c+\theta} |c_i - (c+\theta)|$\label{algline:UpperBound_Restricted_Right_Candidate}\;
\tcp{Place proxies on $c^L$ and $c^R$ to create the $j^\text{th}$ proxy bisector at $b_j \coloneqq (c^L+c^R)/2$}
$p_j^L \leftarrow c^L$ and $p_i^R \leftarrow c^R$\;
$c \leftarrow c^R$\Comment*{Repeat with $c^R$ as the new reference}
$j \leftarrow j+1$
}
\KwRet{$\P = \{p_j^L,p_j^R\}_{j \in [t]}$}\Comment*{$t$ is the number of proxy bisectors computed by the algorithm}
\caption{Algorithm for computing upper bound on no. of proxies under restricted positioning}
\label{alg:UpperBound_Restricted}
\end{algorithm}

The algorithm works by fixing a \emph{reference} candidate $c$ in each iteration (in the first iteration, the leftmost candidate $c_1$ is the reference). Let $c^L \coloneqq \arg\min_{c_i \in \C \, : \, c_i \leq c+\theta} |c_i - (c+\theta)|$ denote the candidate that is weakly to the left of and  closest to the point $c+\theta$. Similarly, let $c^R \coloneqq \arg\min_{c_i \in \C \, : \, c_i > c+\theta} |c_i - (c+\theta)|$ denote the candidate that is strictly to the right of and closest to $c+\theta$. The algorithm places proxies at $c^L$ and $c^R$, thus creating a \emph{proxy bisector} that coincides with the candidate bisector at $(c^L+c^R)/2$. It then proceeds to the next iteration with $c^R$ as the new reference, and this process continues as long as $c+\theta$ remains within $[0,1]$.

It is easy to see that the algorithm runs in polynomial time. Indeed, each candidate can take the role of the reference at most once, and therefore the algorithm performs at most $m$ iterations of the while-loop. Furthermore, each such iteration takes polynomial time since the candidate locations $c_1,\dots,c_m$ and the parameter $\theta$ are assumed to be rational numbers.

Let $\P$ denote the proxy arrangement returned by the algorithm. To see why $\P$ is $\theta$-representative, we will find it convenient to define a partitioning of the line segment $[0,1]$ into \emph{intervals}. Specifically, let $p_j^L$ and $p_j^R$ denote the proxies added by the algorithm in the $j^\text{th}$ iteration, and let $b_j$ denote the corresponding proxy bisector, i.e., $b_j \coloneqq (p_j^L + p_j^R)/2$. Thus, if there are $t$ iterations overall, then $\P = \{p_j^L,p_j^R\}_{j \in [t]}$ and the set of corresponding proxy bisectors is $\{b_1,\dots,b_t\}$.\footnote{Notice that the output $\P$ of Algorithm~\ref{alg:UpperBound_Restricted} can contain up to $2t$ proxies, and thus there can be up to $2t-1$ proxy bisectors between pairs of adjacent proxies. Among these, $b_j$ denotes the proxy bisector between the pair of proxies $\{p_j^L,p_j^R\}$ created in the $j^\text{th}$ iteration.} 

By construction, each proxy bisector $b_j$ coincides with a candidate bisector between adjacent candidates. Therefore, we can associate each candidate $c_i$ with a unique $j \in \{0,1,\dots,t\}$ such that $b_j < c_i < b_{j+1}$, where $b_0 \coloneqq -\infty$ and $b_{t+1} \coloneqq +\infty$. In other words, we obtain a partitioning of the set of candidates $\C = \C_1 \cup \C_2 \cup \dots \cup \C_{t+1}$, where, for every $j \in [t+1]$, $\C_j \coloneqq \{c_i \in \C : b_{j-1} < c_i < b_j\}$ denotes the set of candidates between the proxy bisectors $b_{j-1}$ and $b_j$. 

Recall that the Voronoi cell of candidate $c_i$ is defined as $V_i \coloneqq \{v \in [0,1]: \top(v) = c_i \}$. For every $j \in [t+1]$, we define the \emph{interval} $I_j \coloneqq \cup_{i \in [m] : c_i \in \C_j} V_i$ as the set of locations of all voters whose favorite candidate is in $\C_j$ (recall that ties are broken according to a directionally consistent tie-breaking rule). As before, we will identify all candidates, voters, and proxies with their locations, and will say that `a candidate $c$ (or a proxy $p$) belongs to interval $I_j$' if the voter located at $v=c$ (or $v=p$) is in $I_j$. Notice that $I_1 \cup I_2 \cup \dots \cup I_{t+1} = [0,1]$.

A useful observation is that all candidates in an interval are $\theta$-close. Indeed, the leftmost candidate in each interval plays the role of a reference candidate during some iteration (let us call such a candidate the reference candidate for the interval). By the candidate selection rule in Line~\ref{algline:UpperBound_Restricted_Left_Candidate}, the rightmost candidate in an interval is $\theta$-close to the reference (i.e., the leftmost) candidate in that interval. Thus, we have that for every $j \in [t+1]$, all candidates in $I_j$ (equivalently, all candidates in $\C_j$) are $\theta$-close.

Fix any $j \in [t+1]$ and consider any voter location $v \in I_j$. In order to prove that the proxy arrangement $\P$ is $\theta$-representative, we need to show that $|\top(v) - \top(p^{v})| \leq \theta$. Since $I_j$ is a union of Voronoi cells, it is easy to see that the favorite candidate of voter $v$ must be in $\C_j$, i.e., $\top(v) \in \C_j$. Furthermore, since all candidates in $\C_j$ are $\theta$-close, it suffices to show that voter $v$'s closest proxy also lies in $I_j$ (i.e., $p^v \in I_j$), as that would imply $\top(p^v) \in \C_j$.

In order to prove that $p^v \in I_j$, observe that any voter $v \in I_j$ prefers the leftmost candidate in $\C_j$ over any candidate in $\C_1 \cup \dots \cup \C_{j-1}$. Similarly, it prefers the rightmost candidate in $\C_j$ over any candidate in $\C_{j+1} \cup \dots \cup \C_{t+1}$. Since the proxy bisectors $b_{j-1}$ and $b_j$ coincide with the candidate bisectors, a similar observation holds for the proxies. That is, voter $v$ is closer to the proxy $p^R_{j-1}$ than to any other proxy to its left, and is closer to the proxy $p^L_j$ than to any other proxy to its right. Thus, voter $v$'s closest proxy must be either $p^R_{j-1}$ or $p^L_j$ (i.e., $p^v \in \{p^R_{j-1}, p^L_j\}$), both of which lie in $I_j$, as desired.

Finally, we will show that the total number of proxies in $\P$ satisfies the desired bound. Observe that the total number of bisectors computed by the algorithm is $t < \frac{1}{\theta}$. This is because for any $j \in [t]$, the reference candidates for the intervals $I_j$ and $I_{j+1}$ are separated by a distance strictly greater than $\theta$ (Line~\ref{algline:UpperBound_Restricted_Right_Candidate}), and the distance between the extreme candidates is equal to $1$. If $\theta = \frac{1}{q}$ for some $q \in \mathbb{N}$, we have $t < q$, and since $t$ and $q$ are both integers, we have $t \leq q-1 = \frac{1}{\theta} - 1$. Otherwise, we have $t \leq \lfloor\frac{1}{\theta}\rfloor$. The desired bound now follows by observing that there are at most two distinct proxies per bisector.
\end{proof}

\section{Proof of Proposition~\ref{prop:LowerBound_Restricted}}
\label{subsubsec:Proof_LowerBound_Restricted}

\LowerBoundRestricted*
\begin{proof}
Let $p \in \mathbb{N}$ denote the unique positive integer such that $\frac{1}{p} \leq \theta < \frac{1}{p-1}$. Observe that $p-1$ equals $\frac{1}{\theta} - 1$ when $\frac{1}{\theta} \in \mathbb{N}$, and equals $\lfloor\frac{1}{\theta}\rfloor$ otherwise. Therefore, it suffices to show that any $\theta$-representative arrangement requires $2p-2$ proxies.

Our construction of the lower bound instance will depend on whether $p$ is even or odd. Specifically, let $\eps \coloneqq \frac{1-(p-1)\theta}{5(p-1)/2}$ (if $p$ is odd) or $\eps \coloneqq \frac{1-(p-1)\theta}{(5p/2-3)}$ (if $p$ is even), and observe that $\eps > 0$ in both cases. The distinction between even and odd cases is made in order to ensure that the distance between the extreme candidates is equal to $1$.

The instance consists of the candidate $c_1$ on the left, followed by a sequence of pairs of candidates $\{c_{2i},c_{2i+1}\}$ for $i \in \mathbb{N}$ such that each pair is at a distance of $\theta+\eps$ from the preceding pair, and the gap between the candidates in each pair alternates between $\eps$ and $2\eps$ (see \Cref{fig:LowerBound_Restricted} for an example). In other words, the sequence of distances between adjacent pair of candidates from left to right is given by $\theta+\eps,\eps,\theta+\eps,2\eps,\theta+\eps,\eps$ and so on. The rightmost candidate is $c_{2p-1}$. It is easy to see that all candidate locations are rational, and that the distance between the extreme candidates is equal to $1$.

We will now show that any $\theta$-representative proxy arrangement for the aforementioned instance requires at least $2p-2$ proxies, which, by the above observation, will imply the desired bound. Notice that $c_1$ and $c_2$ are $\theta$-far since $\eps > 0$. Therefore, under any $\theta$-representative proxy arrangement, there must exist a proxy bisector (between adjacent proxies) that coincides with the candidate bisector between $c_1$ and $c_2$ (contrapositive of \Cref{lem:Bisectors_Non_Coincident}). The only way this could happen under restricted positioning is if there are proxies on both $c_1$ and $c_2$. 

Next, observe that the candidates $c_3$ and $c_4$ are also $\theta$-far. Therefore, under any $\theta$-representative proxy arrangement, there must exist a proxy bisector between adjacent proxies that coincides with the candidate bisector between $c_3$ and $c_4$. Notice that the proxy at $c_2$ cannot be one of these proxies, since that would require placing another proxy at the location $c_4+\eps$ which does not have any candidate and would violate restricted positioning assumption. Therefore, there must exist another pair of proxies on $c_3$ and $c_4$.

By a similar argument, the proxy bisector that coincides with the candidate bisector between $c_5$ and $c_6$ cannot correspond to the proxy at $c_4$, since that would involve placing another proxy at $c_6+2\eps$, which does not have any candidate. Therefore, it is necessary to place proxies at $c_5$ and $c_6$.

Continuing in this manner, we observe that the desired arrangement must place proxies on all candidates except for the last one. Thus, the number of proxies required is $2p-2$, as desired.
\end{proof}

\section{Proof of Theorem~\ref{thm:Unrestricted_Proxies_Optimal}}
\label{subsubsec:Proof_Unrestricted_Proxies_Optimal}


%
%
%
%
%
This section provides an algorithm for constructing an optimal $\theta$-representative proxy arrangement. Unlike the restricted case, the set of possible proxy locations is no longer finite, which prohibits the use of simple dynamic programming techniques (as in \Cref{thm:Opt_Restricted}) that iterate over the discrete set of feasible locations. Instead, we must perform a careful analysis of the feasible \emph{regions} of $\theta$-representative proxies, and then prove that such regions can be computed efficiently.

Before discussing the algorithm, let us set up some helpful notation. Recall that for any $i \in [m]$, the Voronoi cell of candidate $c_i \in \C$ is the set of locations of all voters whose favorite candidate is $c_i$, i.e., $V_i = \{v \in [0,1] : \top(v) = c_i\}$. We will find it convenient to define, for every $i \in [m]$, the \emph{extended Voronoi cell} $\myoverbar{V}_i \coloneqq \{v \in (-\infty,\infty) : \top(v) = c_i\}$ as the analog of Voronoi cell where the voters are hypothesized to lie anywhere on the real line. Thus, $\myoverbar{V}_1 = (-\infty,0) \cup V_1$, $\myoverbar{V}_m = V_m \cup (1,\infty)$, and $\myoverbar{V}_i = V_i$ for all $i \in \{2,3,\dots,m-1\}$. Additionally, given any $\theta \in (0,1)$, we will define $\myoverbar{V}_i^+ \coloneqq \cup_{c_\ell \in \C \, : \, c_\ell > c_i+\theta} \myoverbar{V}_\ell$ as the set of all locations of voters (lying anywhere on the real line) whose favorite candidate is $\theta$-far from $c_i$ on the right, and $\myoverbar{V}_i^- \coloneqq \cup_{c_\ell \in \C \, : \, c_\ell < c_i-\theta} \myoverbar{V}_\ell$ as the analogous set on the left side. We will call $\myoverbar{V}^+_i$ and $\myoverbar{V}^-_i$ the \emph{right infeasible} and \emph{left infeasible} regions, respectively, as they denote the regions where the proxy $p^v$ is not allowed to lie if voter $v$'s favorite candidate is $c_i$. Note that the idea of allowing the voters to lie anywhere on the real line is only used for defining the quantities $\myoverbar{V}_i$, $\myoverbar{V}^+_i$, and $\myoverbar{V}^-_i$.

We will also define the Voronoi cells of the proxies. Given a proxy arrangement $\P = \{p_1,\dots,p_k\}$, the Voronoi cell $W_j$ of proxy $p_j$ is the set of all locations of voters whose closest proxy is $p_j$, i.e., $W_j \coloneqq \{v \in [0,1]: |v - p_j| \leq |v - p_\ell| \text{ for any } \ell \neq j\}$, where ties are broken according to a directionally-consistent tie-breaking rule (\Cref{defn:Tie-Breaking}). Notice that $v \in W_j$ if and only if $p_j = p^v$.

Let us start with a characterization of $\theta$-representative proxy arrangements in terms of Voronoi cells of proxies and the left/right infeasible regions that will be useful in our algorithm (\Cref{lem:Unrestricted_Proxies_Characterization}).

\begin{restatable}{lemma}{Characterization}
A proxy arrangement $\P = \{p_1,\dots,p_k\}$ is $\theta$-representative if and only if for every proxy $p_j \in \P$,
$$W_j \cap \{\myoverbar{V}_i^+ \cup \myoverbar{V}_i^-\} = \emptyset,$$
where $c_i \in \C$ is the unique candidate such that $\top(p_j)=c_i$.
\label{lem:Unrestricted_Proxies_Characterization}
\end{restatable}

\begin{proof}
To prove the forward direction, suppose, for contradiction, that for some proxy $p_j \in \P$, there exists a voter $ v \in W_j \cap \{\myoverbar{V}_i^+ \cup \myoverbar{V}_i^-\}$. Then, $\top(v) \in \{c_\ell \in \C : |c_\ell - c_i| > \theta\}$ and $\top(p^v) =  \top(p_j) = c_i$, which violates $\theta$-representation.

In the reverse direction, suppose, for contradiction, that $\P$ is not $\theta$-representative. Then, there must exist a voter $v \in [0,1]$ such that $|\top(v) - \top(p^v)| > \theta$. Let $c_\ell \coloneqq \top(v)$ denote the favorite candidate of voter $v$, $p_j \coloneqq p^v$ denote its closest proxy, and $c_i \coloneqq \top(p_j)$ denote the favorite candidate of $p_j$. Thus, $|c_\ell - c_i| > \theta$, which means that $v \in \myoverbar{V}_i^+ \cup \myoverbar{V}_i^-$. Furthermore, $v \in W(p_j)$. Therefore, $W_j \cap \{\myoverbar{V}_i^+ \cup \myoverbar{V}_i^-\} \neq \emptyset$, which is a contradiction.
\end{proof}

\begin{algorithm*}[t]
\DontPrintSemicolon
\footnotesize
 \linespread{1.1}
\KwIn{A \ProxyVoting{} instance $\I = \langle \C, \theta \rangle$ and a positive integer $k \in \mathbb{N}$.}
\KwOut{YES (if $k$ proxies can form a $\theta$-representative arrangement) or NO (otherwise).}
\BlankLine
\tcp{Initialize the feasibility set $F^1$ for the leftmost proxy}
\BlankLine
\tikzmk{A}
$F^{1,1} \leftarrow \myoverbar{V_1}$\label{algline:Optimal_Unrestricted_Phase1_Start}
\For{$i \in \{2,3,\dots,m\}$}{
    \uIf{$|c_i - c_1| \leq \theta$}{$F^{1,i} = \myoverbar{V}_i$\Comment*{Include $\myoverbar{V}_i$ in the feasibility set if $c_i$ is $\theta$-close to $c_1$}}
    \Else{$F^{1,i} = \emptyset$\Comment*{Discard if $c_i$ is $\theta$-far from $c_1$}}
}
$F^1 \leftarrow \cup_{i \in [m]} F^{1,i}$\label{algline:Optimal_Unrestricted_Phase1_End}\;
\nonl \tikzmk{B}
\boxit{mygray}
\nl
\tcp{Compute the feasibility set $F^j$ for the $j^\text{th}$ proxy $p_j$ for $j \in \{2,\dots,k\}$}
\BlankLine
\tikzmk{A}
Initialize $F^{j,i} \leftarrow \emptyset$ for all $j \in \{2,\dots,k\}$ and $i \in [m]$.\;
\For{$j \in \{2,\dots,k\}$}{
    \For{$h \in [m]$}{
        \For{$i \geq h$}{
            \uIf{$h=i$}{
                $\alpha \leftarrow \inf F^{j-1,h}$\label{algline:Optimal_Unrestricted_Same_Cell_1}\;
                $F^{j,i} \leftarrow F^{j,i} \cup \{ \myoverbar{V}_i \cap (\alpha,\infty) \}$\Comment*{Include all $p_j \in \myoverbar{V}_i$ for which $p_j > p_{j-1}$ for some $p_{j-1} \in F^{j-1,i}$}\label{algline:Optimal_Unrestricted_Same_Cell_2}
            }
            \uElseIf{$\myoverbar{b}_h \geq \myunderbar{b}_i$}
                    {
                    \tcp{Iterate over all maximal convex subsets $S \in F^{j-1,h}$}
                        \For{$S \in F^{j-1,h}$}{
                            $x \leftarrow \inf S$ and $y \leftarrow \sup S$\;
                            $x' \leftarrow \mu(x,\myoverbar{b}_h)$ and $y' \leftarrow \mu(y,\myunderbar{b}_i)$\Comment*{Reflect the endpoints of $S$ about $\myoverbar{b}_h$ and $\myunderbar{b}_i$}
                            $F^{j,i} \leftarrow F^{j,i} \cup \{ [y',x'] \cap \myoverbar{V}_i \}$\Comment*{Compute the intersection of the reflection with $\myoverbar{V}_i$}\label{algline:Optimal_Unrestricted_Reflection}
                        }
                }
        }
    }
    $F^j \leftarrow \cup_{i \in [m]} F^{j,i}$\;
}
\tikzmk{B}
\boxit{mygray}
\uIf{$F^k \cap \{\cup_{c_i \in \C : |c_i-c_m|\leq\theta} \myoverbar{V}_i\} \neq \emptyset$\label{algline:Optimal_Unrestricted_Final_Feasibility_Check}}{
\BlankLine
\KwRet{YES}}
\Else{\KwRet{NO}}
\caption{Algorithm for determining the optimal number of proxies under unrestricted positioning}
\label{alg:Optimal_Unrestricted}
\end{algorithm*}

\paragraph{Description of the algorithm}
Consider a pair of adjacent proxies located at $p_{j-1}$ and $p_j$ such that $\top(p_{j-1}) = c_\ell$ and $\top(p_j) = c_i$. We say that the proxy locations $p_{j-1}$ and $p_j$ are \emph{mutually feasible} if the Voronoi cell $W_{j-1}$ of the proxy $p_{j-1}$ does not overlap with the right infeasible region $\myoverbar{V}^+_\ell$ of the candidate $c_\ell$, and the Voronoi cell $W_j$ of the proxy $p_j$ does not overlap with the left infeasible region $\myoverbar{V}^-_i$ of the candidate $c_i$. From \Cref{lem:Unrestricted_Proxies_Characterization}, it follows that a proxy arrangement is $\theta$-representative if and only if all adjacent pairs of proxies are mutually feasible (in addition to the leftmost and rightmost proxies being feasible for the leftmost and rightmost voters, respectively).

At a high level, our algorithm (Algorithm~\ref{alg:Optimal_Unrestricted}) uses dynamic programming to determine the ``feasibility set'' of each proxy using the feasibility set of its preceding proxy. More concretely, the feasibility set of the leftmost proxy $p_1$ is given by $F^1 \coloneqq \cup_{c_i \in \C \, : \, c_i \leq c_1 + \theta} \myoverbar{V}_i$, which is the set of all locations of $p_1$ for which $\theta$-representation is satisfied for all voters in $[0,p_1]$. For any $j \in \{2,3,\dots,k\}$, the feasibility set $F^j$ of the proxy $p_j$ is the set of all locations of $p_j$ (on the real line) for which there exists some feasible location $p_{j-1} \in F^{j-1}$ of the previous proxy such that $p_j$ and $p_{j-1}$ are mutually feasible. Thus, if all adjacent pairs of proxies in $\{p_1,p_2,\dots,p_j\}$ are mutually feasible, then $\theta$-representation is satisfied for all voters in $W_1 \cup W_2 \cup \dots \cup W_{j-1}$.

Once the feasibility set $F^k$ of the rightmost proxy $p_k$ has been determined, the algorithm checks whether it overlaps with the set $\cup_{c_i \in \C : |c_i-c_m|\leq\theta} \myoverbar{V}_i$. The latter is the set of all locations of the rightmost proxy for which $\theta$-representation is satisfied for the rightmost voter in $[0,1]$. 
If the overlap is non-empty, then there exists a feasible location of the proxy $p_k$ that is $\theta$-representative for the voters in $W_k$. By definition of the feasibility set, there must exist a feasible location $p_{k-1}$ of the previous proxy such that $p_{k-1}$ and $p_k$ are mutually feasible, thus implying $\theta$-representation for voters in $W_{k-1}$. Similarly, there must exist a feasible location of the proxy $p_{k-2}$ that is mutually feasible with $p_{k-1}$, and so on. Continuing backwards in this manner, we obtain a set of proxy locations $p_1,\dots,p_k$ wherein the adjacent pairs are mutually feasible, which immediately implies $\theta$-representation. On the other hand, the absence of an overlap certifies that with $k$ proxies, there is no $\theta$-representative proxy arrangement in the given instance.

\paragraph{Computing the feasibility sets}
For every $j \in [k]$ and every $i \in [m]$, let $F^{j,i} \coloneqq F^j \cap \myoverbar{V}_i$ denote the restriction of the feasibility set of the proxy $p_j$ to the extended Voronoi cell of the candidate $c_i$. Notice that $F^j = \cup_{c_i \in \C} F^{j,i}$. To compute $F^{j,i}$, the algorithm considers $F^{j-1,h}$ for a fixed candidate $c_h \leq c_i$, and computes all locations in $F^{j,i}$ that are mutually feasible with some location in $F^{j-1,h}$. To help with the forthcoming discussion on the computation of $F^{j,i}$, let us set up the relevant notation below.

For every candidate $c_h \in \C$, let $\myoverbar{b}_h$ denote the candidate bisector between the furthest $\theta$-close and the closest $\theta$-far candidates to $c_h$ on the right. That is, $\myoverbar{b}_h \coloneqq \frac{c_\ell + c_{\ell+1}}{2}$ where $c_\ell \coloneqq \arg\min_{c_i \in \C \, : \, c_i \leq c_h+\theta} |c_i - (c_h+\theta)|$. Similarly, let $\myunderbar{b}_h$ denote the analogous bisector on the left side, i.e., $\myunderbar{b}_h \coloneqq \frac{c_{\ell-1} + c_{\ell}}{2}$ where $c_\ell \coloneqq \arg\min_{c_i \in \C \, : \, c_i \geq c_h-\theta} |c_i - (c_h-\theta)|$. In addition, we will define a \emph{mirror} function $\mu: \mathbb{R} \times [0,1] \rightarrow \mathbb{R}$ such that $\mu(x,b)$ denotes the \emph{reflection} of the point $x \in \mathbb{R}$ about the point $b \in [0,1]$; thus, $b = (x + \mu(x,b))/2$.

Let us now discuss the computation of $F^{j,i}$ using $F^{j-1,h}$. Notice that if $h=i$, then any location $x \in \myoverbar{V}_i$ is feasible for proxy $p_j$ as long as there exists some $y \in F^{j-1,h}$ such that $y < x$ (since proxies are assumed to be non-overlapping). Thus, if $\alpha \coloneqq \inf F^{j-1,h}$, then all locations in $\myoverbar{V}_i$ strictly to the right of $\alpha$ are feasible for $p_j$ (Lines~\ref{algline:Optimal_Unrestricted_Same_Cell_1} and~\ref{algline:Optimal_Unrestricted_Same_Cell_2} of Algorithm~\ref{alg:Optimal_Unrestricted}). If $h > i$, then we obtain infeasibility due to $p_{j-1} > p_j$. 

This leaves us with the case $h < i$. Here, we will leverage the characterization in \Cref{lem:Unrestricted_Proxies_Characterization} to compute the desired feasibility set. 
Notice that for a proxy $p_{j-1}$ with $\top(p_{j-1})=c_h$, \Cref{lem:Unrestricted_Proxies_Characterization} immediately implies that the next proxy bisector must be (weakly) to the left of $\myoverbar{b}_h$. Similarly, it implies that for a proxy $p_j$ with $\top(p_{j})=c_i$, the previous proxy bisector must be (weakly) to the right of $\myunderbar{b}_i$. Hence, it cannot be that $\top(p_{j-1})=c_h$ and $\top(p_j)=c_i$ when $\myoverbar{b}_h < \myunderbar{b}_i$.

Thus, we only need to consider the case when $\myoverbar{b}_h \geq \myunderbar{b}_i$. Note that for any $j \in [k]$ and any $i \in [m]$, the set $F^{j,i}$ could, in general, consist of multiple disjoint convex sets that are either open, closed, or semi-open intervals, depending on the tie-breaking rule (we will shortly prove that the number of such sets is polynomially bounded). Let $n^{j,i} \coloneqq |F^{j,i}|$ denote the number of maximal convex subsets of $F^{j,i}$, and let $n^j \coloneqq \sum_{i \in [m]} n^{j,i}$. Note that $n^j$ is an upper bound on the number of maximal convex subsets in $F^j$, since the sets from different $F^{j,i}$'s could be merged when we take their union.

\begin{figure}[h]
\centering
%
\begin{tikzpicture}[line width=0.8pt]
	\tikzset{candidate/.style = {draw,shape=rectangle,scale=1.1}}
	\tikzset{proxy/.style = {circle, draw, red, minimum size=3pt, inner sep=0pt, outer sep=0pt}}
	\tikzmath{\proxyYoffset=-0.3;\proxyXoffset=0.2;}
	    \fill [gray!20] (-7,-0.3) rectangle ++(2.5,0.6);
	    \fill [gray!20] (-1.5,-0.3) rectangle ++(4,0.6);
		\draw[help lines, dashed, step = 0.5cm] (-7.5, -0.45) grid (3, 0.45);
		\node[candidate,fill=blue!20,label=above:{$c_h$}] at (-5,0){};
		\node[candidate] at (-4,0) {};
		\node[candidate] at (-3,0) {};
		\node[candidate] at (-2,0) {};
		\node[candidate,fill=blue!20,label=above:{$c_i$}] at (-1,0) {};
		%
	    \draw[line width=1.2pt, dashed] (-3.5, -0.4) -- (-3.5, 0.4);
	    \node[label=above:{$\myunderbar{b}_i$}] at (-3.5,0.3) {};
	    \draw[line width=1.2pt, dashed] (-2.5, -0.4) -- (-2.5, 0.4);
	    \node[label=above:{$\myoverbar{b}_h$}] at (-2.5,0.3) {};
	    \node[circle,minimum size=1.5pt,inner sep=1pt,fill=red,label=below:{$x$}] (test) at (-7,0) {};
	    \node[circle,minimum size=1.5pt,inner sep=1pt,fill=red,label=below:{$y$}] (test) at (-6.5,0) {};
	    \fill [red] (-6.99,0.02) rectangle (-6.49,-0.02);
	    \node[circle,minimum size=1.5pt,inner sep=1pt,fill=red,label=below:{$y'$}] (test) at (-0.5,0) {};
	    \node[circle,minimum size=1.5pt,inner sep=1pt,fill=red,label=below:{$x'$}] (test) at (2,0) {};
	    \fill [red] (-0.49,0.02) rectangle (1.99,-0.02);
	    %
	    \dimline[extension start length=0,extension end length=0]{(-5,-0.5)}{(-3,-0.5)}{$\theta$};
	    %
	    \dimline[extension start length=0,extension end length=0]{(-3,-0.5)}{(-1,-0.5)}{$\theta$};
\end{tikzpicture}
\caption{Computing the set $F^{j,i}$ using the set $S = [x,y]$ in $F^{j-1,h}$. Here, $x' \coloneqq \mu(x,\overline{b}_h)$ and $y' \coloneqq \mu(y,\underline{b}_i)$. The shaded regions denote the extended Voronoi cells of the candidates. Note that the set $[y',x']$ falls completely inside the extended Voronoi cell $\overline{V}_i$  of candidate $c_i$, and therefore the set $S$ is of Type I.}
\label{fig:Optimal_Unrestricted_Reflection}
\end{figure}
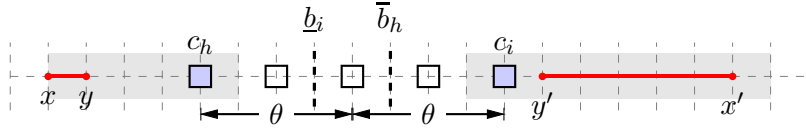

Consider any maximal convex subset $S \in F^{j-1,h}$. For ease of discussion, let us assume that $S$ is a closed interval of the form $[x,y]$ as shown in \Cref{fig:Optimal_Unrestricted_Reflection} (the case of open or semi-open intervals can be handled similarly). By \Cref{lem:Unrestricted_Proxies_Characterization}, any point $z \in F^{j,i}$ that is mutually feasible with some point in $S$ must be (weakly) to the left of the reflection of $x$ about the bisector $\myoverbar{b}_h$, i.e., $z \leq x'$ where $x' \coloneqq \mu(x,\myoverbar{b}_h)$. This is because the proxy bisector between $p_{j-1}$ and $p_j$ has to be (weakly) to the left of $\myoverbar{b}_h$ by \Cref{lem:Unrestricted_Proxies_Characterization}. Also, again by \Cref{lem:Unrestricted_Proxies_Characterization}, for any point $z \in F^{j,i}$, we must have that the reflection of $z$ about the bisector $\myunderbar{b}_i$ is (weakly) to the left of $y$. This is equivalent to asking that $z \geq y'$, where $y' \coloneqq \mu(y,\myunderbar{b}_i)$. Thus, the contribution of the set $S$ to $F^{j,i}$ is given by $\myoverbar{V}_i \cap [y',x']$. The contributions of all maximal convex subsets in $F^{j-1,h}$ to $F^{j,i}$ are taken together as a union in Line~\ref{algline:Optimal_Unrestricted_Reflection} of Algorithm~\ref{alg:Optimal_Unrestricted}.

We will finish this segment by arguing that the number of maximal convex subsets is polynomially bounded. Formally, we will show (via induction) that for every $j \in [k]$, $n^j = \sum_{i \in [m]} n^{j,i} \leq j \cdot 2m$. The base case of $j=1$ is easy to verify: Indeed, for every $i \in [m]$, the set $F^{1,i}$ is either empty or consists of a single maximal convex subset, namely the extended Voronoi cell of candidate $c_i$ (Lines~\ref{algline:Optimal_Unrestricted_Phase1_Start} to~\ref{algline:Optimal_Unrestricted_Phase1_End} in Algorithm~\ref{alg:Optimal_Unrestricted}). Thus, for every $i \in [m]$, $n^{1,i} \leq 1$, and therefore $n^1 \leq m$.

Now suppose $n^{j-1} \leq (j-1) \cdot 2m$. We want to show that $n^j \leq j \cdot 2m$. Fix some $h \in [m]$ and some $i \geq h$, and consider any fixed maximal convex subset $S \in F^{j-1,h}$. Observe that $S$ is of one of the following two types: If the interval $[y',x']$ is completely contained inside the extended Voronoi cell $\myoverbar{V}_i$, we say that $S$ is of \emph{Type I} for $i$ (see \Cref{fig:Optimal_Unrestricted_Reflection}). Otherwise, if $[y',x']$ overlaps partially with $\myoverbar{V}_i$ or not at all, then we say that it is \emph{Type II} for $i$ (see \Cref{fig:Optimal_Unrestricted_Type_II}). 

\begin{figure}[h]
\centering
%
\begin{tikzpicture}[line width=0.8pt]
	\tikzset{candidate/.style = {draw,shape=rectangle,scale=1.1}}
	\tikzset{proxy/.style = {circle, draw, red, minimum size=3pt, inner sep=0pt, outer sep=0pt}}
	\tikzmath{\proxyYoffset=-0.3;\proxyXoffset=0.2;}
	    \fill [gray!20] (-7,-0.3) rectangle ++(2.5,0.6);
	    \fill [gray!20] (-1.5,-0.3) rectangle ++(2.5,0.6);
		\draw[help lines, dashed, step = 0.5cm] (-7.5, -0.45) grid (3, 0.45);
		\node[candidate,fill=blue!20,label=above:{$c_h$}] at (-5,0){};
		\node[candidate] at (-4,0) {};
		\node[candidate] at (-3,0) {};
		\node[candidate] at (-2,0) {};
		\node[candidate,fill=blue!20,label=above:{$c_i$}] at (-1,0) {};
		%
	    \draw[line width=1.2pt, dashed] (-3.5, -0.4) -- (-3.5, 0.4);
	    \node[label=above:{$\myunderbar{b}_i$}] at (-3.5,0.3) {};
	    \draw[line width=1.2pt, dashed] (-2.5, -0.4) -- (-2.5, 0.4);
	    \node[label=above:{$\myoverbar{b}_h$}] at (-2.5,0.3) {};
	    \node[circle,minimum size=1.5pt,inner sep=1pt,fill=red,label=below:{$x$}] (test) at (-7,0) {};
	    \node[circle,minimum size=1.5pt,inner sep=1pt,fill=red,label=below:{$y$}] (test) at (-6.5,0) {};
	    \fill [red] (-6.99,0.02) rectangle (-6.49,-0.02);
	    \node[circle,minimum size=1.5pt,inner sep=1pt,fill=red,label=below:{$y'$}] (test) at (-0.5,0) {};
	    \node[circle,minimum size=1.5pt,inner sep=1pt,fill=red,label=below:{$x'$}] (test) at (2,0) {};
	    \fill [red] (-0.49,0.02) rectangle (1,-0.02);
	    \fill [red!20] (1,0.02) rectangle (1.99,-0.02);
	    %
	    \dimline[extension start length=0,extension end length=0]{(-5,-0.5)}{(-3,-0.5)}{$\theta$};
	    %
	    \dimline[extension start length=0,extension end length=0]{(-3,-0.5)}{(-1,-0.5)}{$\theta$};
\end{tikzpicture}
\caption{Illustrating a Type II set $[x,y]$. Observe that the intersection $[y',x'] \cap \overline{V}_i$ is right extremal for $\overline{V}_i$.}
\label{fig:Optimal_Unrestricted_Type_II}
\end{figure}
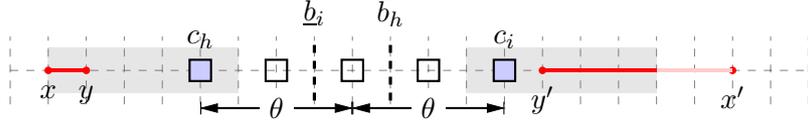

Note that the reflection $x' \coloneqq \mu(x,\myoverbar{b}_h)$ is independent of $i$, and can lie in at most one $\myoverbar{V}_i$. Therefore, each set $S \in F^{j-1,h}$ can be Type I for at most one $i$, thereby contributing at most one maximal convex subset to $F^j$. If $S = [x,y]$ is Type II for $i$, then the intersection $[y',x'] \cap \myoverbar{V}_i$ must be either \emph{left extremal} or \emph{right extremal} in $\myoverbar{V}_i$. A priori, it may seem that each such set $S$ contributes $\O(m)$ convex subsets to $F^{j}$, which might make the cardinality $n^j$ grow exponentially in $j$. However, the key observation is that when taking the union in Line~\ref{algline:Optimal_Unrestricted_Reflection}, \emph{all} left/right extremal subsets within each candidate's extended Voronoi cell are \emph{merged}, and thus we never have more than two extremal sets per candidate. Therefore, the combined contribution of all Type II sets $S \in F^{j-1,h}$ for all $h \in [m]$ towards $F^{j,i}$ (where $i \in [m]$ is fixed) results in an additive increase of at most $2$ in the cardinality of $F^{j,i}$, which, in turn, results in an additive increase of at most $2m$ in the cardinality of $F^j$. Thus, $n^j \leq n^{j-1} + 2m \leq (j-1) \cdot 2m + 2m = j \cdot 2m$, as desired.

With the above setup in hand, let us now prove our main result.

\UnrestrictedProxiesOptimal*
\begin{proof}
We will show that Algorithm~\ref{alg:Optimal_Unrestricted} computes the desired proxy arrangement.

Let us start with the running time analysis. As shown previously, for every $j \in [k]$, the number of maximal convex subsets in each feasibility set $F^j$ is at most $j \cdot 2m$. Furthermore, each such set can be computed efficiently, since each reflection or intersection operation takes polynomial time (note that all maximal convex subsets, which are either closed, open, or semi-open intervals, can be stored in terms of the end-points of their closures, and all end-points are rational). Thus, the overall running time is polynomial.

Notice that although the pseudocode of Algorithm~\ref{alg:Optimal_Unrestricted} is shown to return only a YES/NO outcome, one can actually compute a feasible proxy arrangement (if it exists) as follows: Pick \emph{any} location in $F^k \cap \{\cup_{c_i \in \C : |c_i-c_m|\leq\theta} \myoverbar{V}_i\}$ for the proxy $p_k$ (such a location must exist by the feasibility check in Line~\ref{algline:Optimal_Unrestricted_Final_Feasibility_Check}). Then, compute a mutually feasible location of the proxy $p_{k-1}$. This can be computed by a procedure similar to the one used in computing $F^{j,i}$ using $F^{j-1,h}$. Next, the location of the proxy $p_{k-2}$ is computed, and so on. Finally, the $\theta$-representation of the proxy arrangement returned by the algorithm follows from the mutual feasibility condition maintained by the algorithm in each iteration.
\end{proof}

\section{Proof of Theorem~\ref{thm:UpperBound_Unrestricted}}
\label{subsubsec:Proof_UpperBound_Unrestricted}

\begin{algorithm*}
\DontPrintSemicolon
\footnotesize
 \linespread{1.1}
\KwIn{A \ProxyVoting{} instance $\I = \langle \C, \theta \rangle$.}
\KwOut{A $\theta$-representative arrangement of proxies.}
\BlankLine
\tcp{\scriptsize{Phase 1: Compute locations of proxy bisectors}}
\tikzmk{A}
Initialize $c \leftarrow c_1$\Comment*{Start with $c_1$ as the reference candidate}
Initialize $j \leftarrow 1$\Comment*{$j$ is the index of the proxy bisector}
\While{$c + \theta < 1$\label{algline:UpperBound_Unrestricted_Location_Check}}{
\tcp{$c^L$ and $c^R$ are the candidates closest to $(c+\theta)$ on the left and the right side, respectively}
$c^L \leftarrow \arg\min_{c_i \in \C \, : \, c_i \leq c+\theta} |c_i - (c+\theta)|$\;
$c^R \leftarrow \arg\min_{c_i \in \C \, : \, c_i > c+\theta} |c_i - (c+\theta)|$\label{algline:UpperBound_Unrestricted_Bisector_Selection}\;
$b_i = (c^L+c^R)/2$\label{algline:UpperBound_Unrestricted_Bisector_Location}\Comment*{Bisector of $c^L$ and $c^R$}
$c \leftarrow c^R$\Comment*{Repeat with $c^R$ as the new reference}
$j \leftarrow j+1$
}
\nonl \tikzmk{B}
\boxit{mygray}
\tcp{\scriptsize{Phase 2: Compute proxy locations}}
\oldnl \tikzmk{A}
$B \leftarrow \{b_1,b_2,\dots,b_t\}$\Comment*{Set of proxy bisectors computed in Phase 1}
\For{$j \in [t-1]$}{$w_j \leftarrow b_{j+1} - b_{j}$\Comment*{Compute widths of intervals between consecutive bisectors}}
$W \leftarrow \{w_1,w_2,\dots,w_{t-1}\}$\Comment*{Set of interval widths}
\BlankLine
\tcp{Initialize a proxy arrangement with two proxies per bisector}
Initialize $\eps \leftarrow \frac{1}{3} \cdot \min_{i \in [m-1]} |c_{i+1} - c_i|$\;
\For{$j \in [t]$}{
	$p_j^L = b_j - \eps$\label{algline:UpperBound_Unrestricted_Proxy_Initialization_L}\;
	$p_j^R = b_j + \eps$\label{algline:UpperBound_Unrestricted_Proxy_Initialization_R}
}
\BlankLine
\tcp{Move $p_j^L$ and $p_j^R$ away from the bisector $b_j$ in opposite directions until a ``collision'' occurs}
\For{$\ell \in [t-1]$}{
	$\Delta \leftarrow 0.5 \cdot (\min^{(\ell)} W - \min^{(\ell-1)} W)$\;\Comment{The step size $\Delta$ is chosen as half of the difference between the $\ell^\text{th}$ and $(\ell-1)^\text{th}$ smallest elements in $W$ (i.e., difference in the widths of the $\ell^\text{th}$ and $(\ell-1)^\text{th}$ smallest intervals); here, $\min^{(0)} W \coloneqq 2\eps$.}
	\BlankLine
	\tcp{Move the left and right proxies in opposite directions by a distance $\Delta$}
	\For{$j \in [t]$}{
		\uIf{$b_j$ is not ``frozen''}{
			$p_j^L \leftarrow p_j^L - \Delta$\label{algline:UpperBound_Unrestricted_Expansion_L}\;
			 $p_j^R \leftarrow p_j^R + \Delta$\label{algline:UpperBound_Unrestricted_Expansion_R}}
	}
	\BlankLine
\tcp{Check for ``collision'' events}
	\For{$j \in [t-1]$}{
		\uIf{$p_j^R = p_{j+1}^L$\label{algline:UpperBound_Unrestricted_Proxy_Collision}}{Fix the locations of $p_j^L$, $p_j^R$, $p_{j+1}^L$, and $p_{j+1}^R$\;
		Mark the proxy bisectors $b_j$ and $b_{j+1}$ as ``frozen''}
	}
}
\tikzmk{B}
\boxit{mygray}
\KwRet{$\P = \{p_j^L,p_j^R\}_{j \in [t]}$}
\caption{Algorithm for computing upper bound on no. of proxies under unrestricted positioning}
\label{alg:UpperBound_Unrestricted}
\end{algorithm*}

\UpperBoundUnrestricted*
\begin{proof}
We will show that Algorithm~\ref{alg:UpperBound_Unrestricted} computes the desired proxy arrangement.

Let us start by showing that Algorithm~\ref{alg:UpperBound_Unrestricted} terminates in polynomial time. Observe that in Phase 1, the algorithm performs at most $m$ iterations of the while-loop (since a candidate can be a reference at most once) and each such iteration takes polynomial time (since the candidate locations $c_1,\dots,c_m$ and the parameter $\theta$ are assumed to be rational numbers). Furthermore, since each proxy bisector $b_j$ overlaps with a candidate bisector between adjacent candidates (Line~\ref{algline:UpperBound_Unrestricted_Bisector_Location}), we know that there are at most $m-1$ proxy bisectors in total (i.e., $t \leq m-1$) and that each $b_j$ is rational. The latter implies that all interval widths in the set $W$ as well as the quantity $\Delta$ in each iteration are also rational. Thus, each for-loop iteration in Phase 2 takes polynomial time, and since there can be at most $\O(t^2)$ such iterations, the desired running time guarantee follows.

Let $\P = \{p_j^L,p_j^R\}_{j \in [t]}$ denote the proxy arrangement returned by Algorithm~\ref{alg:UpperBound_Unrestricted}. We will now argue that $\P$ is $\theta$-representative. To prove this, it would be helpful to recall the notion of an ``interval'' from the proof of \Cref{thm:UpperBound_Restricted}. Let $b_1,\dots,b_t$ denote the proxy bisectors computed by the algorithm in Phase 1, where $b_j \coloneqq (p_j^L + p_j^R)/2$. As noted in the proof of \Cref{thm:UpperBound_Restricted}, the candidates can be partitioned into the sets $\C_1,\C_2,\dots,\C_{t+1}$, where $\C_j \coloneqq \{c_i \in \C : b_{j-1} < c_i < b_j\}$ for $j \in [t+1]$, assuming $b_0 \coloneqq -\infty$ and $b_{t+1} \coloneqq +\infty$. For every $j \in [t+1]$, the \emph{interval} $I_j \coloneqq \cup_{i \in [m] : c_i \in \C_j} V_i$ is the set of locations of all voters whose favorite candidate is in $\C_j$. It is easy to verify that all candidates within an interval are $\theta$-close.

Consider any voter location $v \in [0,1]$ such that $v \in I_j$. Thus, $\top(v) \in \C_j$. 
In order to show that the proxy arrangement $\P$ is $\theta$-representative, it suffices to establish that voter $v$'s closest proxy $p^v$ is such that $p^v \in \{p_j^R, p_{j+1}^L\}$. Before proving this claim, let us argue why it gives the desired result. Notice that the initialization in Lines~\ref{algline:UpperBound_Unrestricted_Proxy_Initialization_L} and~\ref{algline:UpperBound_Unrestricted_Proxy_Initialization_R} and the expansion operation in Lines~\ref{algline:UpperBound_Unrestricted_Expansion_L} and~\ref{algline:UpperBound_Unrestricted_Expansion_R} together ensure that $p_j^R \geq b_j + \eps$ and $p_{j+1}^L \leq b_{j+1} - \eps$ at every step. Furthermore, the collision check in Line~\ref{algline:UpperBound_Unrestricted_Proxy_Collision} ensures that $p_j^R \leq p_{j+1}^L$. (Note that the reason the proxies $p_j^R$ and $p_{j+1}^L$ do not ``cross over'' is because the step size $\Delta$ is chosen as the \emph{difference} in widths of the intervals.) Thus, overall, we have that $p_j^R, p_{j+1}^L \in [b_j + \eps, b_{j+1} - \eps]$, which means that the proxies $p_j^R$ and $p_{j+1}^L$ are contained in the interval $I_j$. This, in turn, implies that the favorite candidates of these proxies, namely $\top(p_j^R)$ and $\top(p_{j+1}^L)$, are contained in $\C_j$. Since all candidates within $\C_j$ are $\theta$-close, we get that $|\top(v) - \top(p^v)| \leq \theta$, as desired.

The reason why $p^v \in \{p_j^R, p_{j+1}^L\}$ holds is also similar to that in the proof of \Cref{thm:UpperBound_Restricted}. Indeed, voter $v$ prefers the leftmost candidate in $\C_j$ over any candidate in $\C_1 \cup \dots \cup \C_{j-1}$, and prefers the rightmost candidate in $\C_j$ over any candidate in $\C_{j+1} \cup \dots \cup \C_{t+1}$. Since the proxy bisectors coincide with the candidate bisectors, it follows that voter $v$ is closer to the proxy $p^R_{j-1}$ than to any other proxy to its left, and is closer to the proxy $p^L_j$ than to any other proxy to its right. Thus, voter $v$'s closest proxy must be either $p^R_{j-1}$ or $p^L_j$, i.e., $p^v \in \{p^R_{j-1}, p^L_j\}$.

We will now show that the total number of proxies is $|\P| \leq \frac{3}{2}\lceil\frac{1}{\theta}\rceil$. First, observe that the total number of bisectors computed by the algorithm is $t \leq \lfloor \frac{1}{\theta} \rfloor$. This is because for any $j \in [t]$, the reference candidates for the intervals $I_j$ and $I_{j+1}$ are separated by a distance strictly greater than $\theta$ (Line~\ref{algline:UpperBound_Unrestricted_Bisector_Selection}), and the distance between the extreme candidates is equal to $1$. Next, observe that there is no pair of consecutive intervals $I_j$ and $I_{j+1}$ with two proxies each, since at least one of these four proxies will be absorbed during the expansion step in Phase 2 (Lines~\ref{algline:UpperBound_Unrestricted_Expansion_L} and~\ref{algline:UpperBound_Unrestricted_Expansion_R}). Therefore, any pair of adjacent intervals together contain at most three proxies, resulting in the bound $\frac{3}{2}\lfloor\frac{1}{\theta}\rfloor$ on the total number of proxies. Finally, in the special case when there is only one bisector (i.e., when $t=1$), we have exactly two proxies, and the bound $\frac{3}{2}\lceil\frac{1}{\theta}\rceil$ is satisfied for any given $\theta \in (0,1)$.
\end{proof}

Although \Cref{thm:UpperBound_Unrestricted} provides an absolute upper bound on the number of proxies, this number could, in general, be much larger than the \emph{optimal} number of proxies for the instance at hand. Indeed, in an instance with two candidates located at the endpoints of $[0,1]$ and $\theta=0.01$, the bound in \Cref{thm:UpperBound_Unrestricted} evaluates to $150$, even though it is clear that two proxies suffice. \Cref{cor:Approx_Algo_Unrestricted} shows that Algorithm~\ref{alg:UpperBound_Unrestricted} never uses more than three times the optimal number of proxies for the given instance. 

\begin{restatable}[]{corollary}{ApproxAlgoUnrestricted}
Algorithm~\ref{alg:UpperBound_Unrestricted} is a $3$-approximation algorithm for \ProxyVoting{}.
\label{cor:Approx_Algo_Unrestricted}
\end{restatable}

For any fixed $\alpha \in [0,1]$, we say that an algorithm is an \emph{$\alpha$-approximation algorithm} for \ProxyVoting{} if, given any instance $\I$ as input, it terminates in polynomial time and returns a $\theta$-representative arrangement of at most $\alpha \cdot \opt(\I)$ proxies.

\begin{proof} (of \Cref{cor:Approx_Algo_Unrestricted})
Recall 
that the Voronoi cell $V_i$ of candidate $c_i \in \C$ is the set of all voter locations $v \in [0,1]$ for which $\top(v) = c_i$. Let us define the \emph{segment} $S_i$ associated with the candidate $c_i$ as the union of Voronoi cells of all candidates that are $\theta$-close to $c_i$, i.e., $S_i \coloneqq \cup_{\ell \in [m]:|c_\ell-c_i| \leq \theta} V_\ell$. It is easy to see that if a voter is located at $c_i$ (i.e., $v = c_i$), then under a $\theta$-representative proxy arrangement, its closest proxy must be located in $S_i$ (i.e., $p^v \in S_i$).

Let $\I = \langle \C, \theta \rangle$ denote a given instance of \ProxyVoting{}, and let $\opt(\I)$ denote the optimal number of proxies in any $\theta$-representative proxy arrangement in $\I$. Let $c_{i_1},c_{i_2},\dots,c_{i_t}$ denote the set of candidates from left to right that are chosen as reference candidates during Phase 1 of Algorithm~\ref{alg:UpperBound_Unrestricted}. Additionally, let $S_{i_1},\dots,S_{i_t}$ denote the corresponding segments. 

We will argue that the alternate segments, namely $S_{i_1},S_{i_3},\dots,S_{i_t}$ (if $t$ is odd) or $S_{i_1},S_{i_3},\dots,S_{i_{t-1}}$ (if $t$ is even), are disjoint. (Stated differently, the set of alternate segments constitutes an independent set of the underlying interval graph.) Note that this would imply that $\opt(\I) \geq \lceil\frac{t}{2}\rceil$ since, by the above observation, any $\theta$-representative proxy arrangement must have a proxy in each segment. We know from the proof of \Cref{thm:UpperBound_Unrestricted} that the total number of proxies under Algorithm~\ref{alg:UpperBound_Unrestricted} is at most $\frac{3}{2}$ times the number of proxy bisectors computed by the algorithm in Phase 1. Thus, with $t$ bisectors, Algorithm~\ref{alg:UpperBound_Unrestricted} uses at most $\frac{3}{2}t$ proxies, which gives the stated approximation guarantee.

To see why the alternate segments are disjoint, fix any $\ell \in [t]$ and consider the segments $S_{i_\ell}$, $S_{i_{\ell+1}}$, and $S_{i_{\ell+2}}$. Since the candidates $c_{i_\ell}$ and $c_{i_{\ell+1}}$ are chosen as reference candidates in consecutive iterations, we have that $c_{i_{\ell+1}} > c_{i_\ell} + \theta$, and therefore $c_{i_{\ell+1}}$ (and any candidate to its right) is not included in the segment $S_{i_\ell}$. For a similar reason, the candidate $c_{i_{\ell+1}}$ (and any candidate to its left) is excluded from the segment $S_{i_{\ell+2}}$, implying that $S_{i_\ell}$ and $S_{i_{\ell+2}}$ are disjoint, as desired.
\end{proof}

\section{Proof of Proposition~\ref{prop:rep-voting}}
\label{sec:Proof_rep-voting}

\Condorcet*
\begin{proof}
We first prove the proposition when $r$ is a strict-Condorcet rule, where the tie-breaking mechanism favors the leftmost alternative. The proof for the other tie-breaking mechanism is similar. 

For any preference profile $P$ with odd number of voters, let $\bar v$ denote the location of $P$'s median voter(s). Let $p^P = \{p^v:v\in P \}$ denote the proxy profile. It is well-known that the top choice of $\bar v$ is the Condorcet winner (see, e.g.,~\citealp{Congleton2002:The-Median}), which means that $r(P)$ is the candidate that is closest to $\bar v$, i.e.~$\top(\bar v)$. Therefore, it suffices to prove that $r(p^P) = \top(p^{\overline v})$, or in other words, $p^{\bar v}$ is the median of $p^P$. This follows after the fact that for any voter $v'\in P$ such that $v'\le \bar v$ (respectively, $v'\ge \bar v$), we must have $p^{v'}\le p^{\bar v}$ (respectively, $p^{v'}\ge p^{\bar v}$).

For any preference profile $P$ with even number of voters, it is possible to have two median voters, whose locations are denoted by $\bar v_1<\bar v_2$, respectively. It follows that the weak Condorcet winners are $\{\top(\bar v_1), \top(\bar v_2)\}$, and in case the set contains more than one candidate, according to the tie-breaking mechanism, the leftmost candidate $\top(\bar v_1)$ wins. Therefore, it suffices to prove that $r(p^P) = \top(p^{\bar v_1})$. Similar to the case with odd number of voters, it is not hard to verify that the median locations of $p^{P}$ are $\{p^{\bar v_1},p^{\bar v_2}\}$. In case the set contains two locations, according to the tie-breaking mechanism, $\top(p^{\bar v_1})$ will be chosen as the winner. This completes the proof.
\end{proof}

\end{document}